\def\year{2021}\relax
\documentclass[letterpaper]{article} %
\usepackage{aaai21}  %
\usepackage{times}  %
\usepackage{helvet} %
\usepackage{courier}  %
\usepackage[hyphens]{url}  %
\usepackage{graphicx} %
\usepackage{natbib}  %
\usepackage{caption} %
\title{My Publication Title --- Multiple Authors}
\author {
        Runsheng Yu,\textsuperscript{\rm 1}
        Yu Gong, \textsuperscript{\rm 2}
        Xu He, \textsuperscript{\rm 1}
        Bo An,\textsuperscript{\rm 1}
        Yu Zhu
       , \textsuperscript{\rm 2}
        Qingwen Liu,\textsuperscript{\rm 2}
        Wenwu Ou
        \textsuperscript{\rm 2}
        \\
}
\affiliations {
    \textsuperscript{\rm 1} 
School of Computer Science and Engineering, Nanyang Technological University \\
    \textsuperscript{\rm 2} Alibaba Group \\
    runshengyu@gmail.com, gongyu.gy@alibaba-inc.com, hexu0003@e.ntu.edu.sg, boan@ntu.edu.sg,  zy143829@alibaba-inc.com, xiangsheng.lqw@alibaba-inc.com, santong.oww@taobao.com
}
\usepackage[utf8]{inputenc} %
\usepackage{booktabs}       %
\usepackage{amsfonts}       %
\usepackage{nicefrac}       %
\usepackage{microtype}      %

\usepackage{amsmath}
\usepackage[linesnumbered,ruled,vlined]{algorithm2e}
\usepackage{mathrsfs,dsfont}
\usepackage{amsthm,amssymb}
\usepackage{lipsum}
\usepackage{wrapfig}
\usepackage{bm}
\usepackage{float} 

\newtheorem{remark}{Remark}
\newtheorem{lemma}{Lemma}

\newtheorem{proposition}{Proposition}

\SetKwInput{KwInput}{Input}
\SetKwInput{KwOutput}{Output}
\SetKwFunction{KwFn}{Fn} 
\usepackage{subfigure}

\usepackage{multirow}
\usepackage{bbold}
\usepackage{listings}
\usepackage{color}
\usepackage{pythonhighlight}
\usepackage{lipsum}
\usepackage[switch]{lineno}

\newcommand{\rs}{\textcolor{black} }

\title{Personalized Adaptive Meta Learning for Cold-start User Preference Prediction}

\begin{document}
\maketitle

\begin{abstract} 
A common challenge in personalized user preference prediction is the cold-start problem. Due to the lack of user-item interactions, directly learning from the new users' log data causes serious over-fitting problem. Recently, many existing studies regard the cold-start personalized preference prediction as a few-shot learning problem, where each user is the task and recommended items are the classes, and the gradient-based meta learning method (MAML) is leveraged to address this challenge. However, in real-world application, the users are not uniformly distributed (i.e., different users may have different browsing history, recommended items, and user profiles. We define the major users as the users in the groups with large numbers of users sharing similar user information, and other users are the minor users), existing MAML approaches tend to fit the major users and ignore the minor users. To address this cold-start task-overfitting problem, we propose a novel personalized adaptive meta learning approach to consider both the major and the minor users with three key contributions: 1) We are the first to present a personalized adaptive learning rate meta-learning approach to improve the performance of MAML by focusing on both the major and minor users. 2) To provide better personalized learning rates for each user, we introduce a similarity-based method to find similar users as a reference and a tree-based method to store users' features for fast search. 3) To reduce the memory usage, we design a memory agnostic regularizer to further reduce the space complexity to constant while maintain the performance. Experiments on MovieLens, BookCrossing, and real-world production datasets reveal that our method outperforms the state-of-the-art methods dramatically for both the minor and major users. 
\end{abstract}

\section{Introduction}

\rs{Recommender Systems (RS) help people to discover the items they prefer~\citep{guo2017deepfm,qu2016product}. In order to train a well-performing personalized user preference predictor, enough interactions with users are indispensable. %
To address this challenge, many researchers take advantage of the offline supervised training methods, which leverage the historical data (log data) to train the model. 
However, when the log data is lacking, these  supervised training methods cause the over-fitting issues~\citep{vanschoren2018meta}, which is known as the \emph{cold-start problem}.}

To train a well-performing model in the cold-start problem, meta learning-based approaches are introduced~\citep{lee2019melu,dong2020mamo}. Intuitively, learning with a few samples can be viewed as a few-shot learning problem and meta learning, especially gradient-based meta learning (e.g., Model Agnostic Meta Learning, MAML~\citep{finn2017model}), which aims to adapt to any task with only a few steps of parameter update, has been proven as one of the most successful approaches for these problems. Thanks to its good generalization ability, MAML has already been leveraged into various domains, including computer vision, natural language process, and robotics~\citep{gui2018few,madotto2019personalizing,finn2017one}. Recently, in the RS area, meta learning is introduced for the cold-start problem for either users or items, which treats the users/items as tasks, log data as samples, and learns to do fast adaptation when meeting new tasks~\citep{dong2020mamo, Pan:2019:WUC:3331184.3331268,lee2019melu,luo2020metaselector}. 

\begin{figure}
  \includegraphics[width=\linewidth]{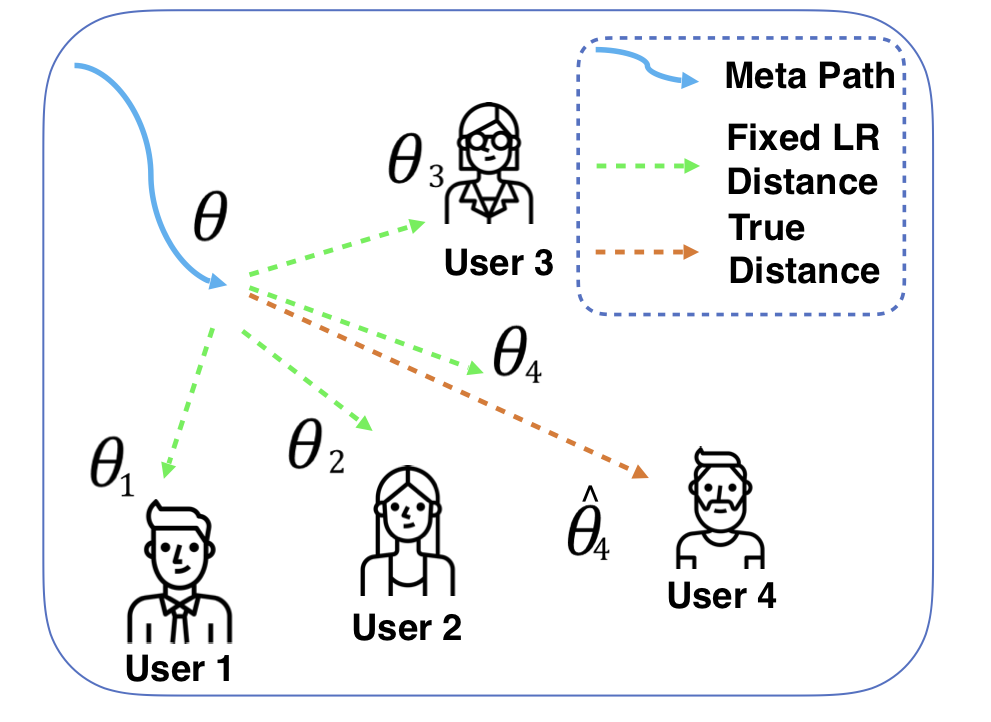}
  \caption{An example illustrating that MAML has limitations in the cold-start user imbalanced dataset. Users $1,2,3$ are the major users and user $4$ is the minor user.} %
  \label{meta_reason}
\end{figure}
\rs{However, these MAML methods assume that the distribution of tasks (users) is uniform. But in the recommendation systems, the user distribution is not always balanced, i.e., the user profile, browsing history, as well as the features of recommended items are not always uniformly distributed\footnote{They are collectively named as the feature values.}, which harms the performance of MAML. For example, regarding the user profile, as shown in Tab.~\ref{ratio}, most values are imbalanced. In terms of the browsing histories, the distribution of the browsing histories is also non-uniform~\citep{Pan:2019:WUC:3331184.3331268}.} Therefore, simply learning from these data will overfit the users with similar feature values since these users own the major similar features and fitting these users can already achieve good-enough \emph{average} prediction performance (like the MSE results in Tab.~\ref{ratio}), which is known as the cold-start user (task) overfitting problem. \rs{To illustrate the problem clearly, we regard a collection of users with similar feature values as a group and define the users as the major users when the number of users in these groups are large, and others users are the minor users.} \rs{Now, we give an example to illustrate how imbalanced distribution harms the performance of MAML: as shown in Fig.~\ref{meta_reason}, assuming that an MAML strategy aims to learn to do fast adaptation to the four cold-start users (\rs{three major users (users 1-3) and one minor user (user 4)}). Different locations in the blue square indicate feature values (embeddings) for different users. Since three of them are near each other (have similar feature values), %
the model may \emph{not} focus on the minor user (user 4) because fitting the major users can already achieve good prediction performance. } %

To address these limitations, in this paper, we propose a personalized adaptive meta learning approach to improve the performance of MAML in the cold-start user preference prediction problem. The main idea is to set different learning rates for different users so as to find task-adaptive parameters for each user. Our main contributions are as follows: 1) We are the first to introduce a novel user-adaptive learning rate based meta learning in RS to improve the performance of MAML by focusing on both the major and minor users. 2) To provide better-personalized learning rates, we introduce a method to find similar users as a reference and a tree-based method to store users' features for fast search.
3) To reduce the memory usage, we introduce a memory agnostic regularizer to reduce the space complexity while maintaining good prediction performance. Experiments on MovieLens-1M, BookCrossing, and real-world production datasets from one of the largest e-commerce platforms reveal that our method outperforms state-of-the-art methods dramatically for both the minor and the major users. 

\section{Related Works}
\begin{table}
{\scriptsize
\begin{tabular}{|c|c|c|c|c|c|c|c|}
\toprule
 \multicolumn{4}{|c|}{\textbf{The ratio of features}} & \multicolumn{2}{c|}{\textbf{The major users}} & \multicolumn{2}{c|}{\textbf{The  minor users}} \\ \hline
Age & Gender & Zipcode & Occup  & Ratio & MSE & Ratio & MSE \\ \midrule
72.4$\%$ & 71.7$\%$ & 56.7$\%$ &  64.0$\%$ & 0.648 & 1.413 & 0.352 &  1.506  \\ 
\bottomrule
\end{tabular}
}
\caption{The ratio \rs{($\frac{\text{the number of users having certain feature values}}{\text{total number of users}}$)} of the top $30\%$ largest number of feature values that users own\footnote{This feature only has two values (male/female) and we use top 50$\%$ instead.}, ratio of certain users ($\frac{\text{the number of users in certain group}}{\text{total number of users}}$), and the Mean Squared Errors (MSE) from one of the MAML methods~\citep{lee2019melu} in MovieLens. The method used to split the cold-start major and minor users can be found in Sec. \ref{exps}.}
\label{ratio}
\end{table}

In this section, we discuss some related works including gradient-based meta learning for the imbalanced dataset and the meta learning for cold-start recommendation.. 

\textbf{Gradient-Based Meta Learning for Task Overfitting Problem.} MAML based methods %
have been widely adopted for studying the few-shot learning problem~\citep{finn2017model,li2017meta,xu2018meta,chen2018federated,ravi2016optimization,lee2018gradient}. %
To consider task-adaptive challenges, the vector of learning rates, the block-diagonal preconditioning matrix, latent embedding, and interleaving warp-layers are designed to learn parameters for different tasks respectively~\citep{li2017meta,park2019meta,rusu2018meta,flennerhag2019meta}. However, those methods still do not explicitly consider the task overfitting problem. To take this problem into consideration, Bayesian
TAML introduces different initiation parameters and inner gradient approaches to handle class imbalance, task imbalance, as well as out-of-distribution tasks together by higher-order statistics~\citep{flennerhag2019meta}. However, in RS, a users' feature space is high-dimensional and those statistics are still hard to represent the embedding completely. \rs{ As far as we know, we are the \emph{first} to propose a novel meta learning method focusing on the two key challenges in cold-start RS: 1) 
comparing with the common-seen few-shot learning problem which does not explicitly have the task information, the dimension of features for the task information (like the user profile)  in our problem is large, and 2) task does not obey the uniform distribution (even long-tail) while the standard few-shot learning problem is assumed uniformly distributed. }

\textbf{Meta Learning for Cold-start Recommendation.} In RS, meta learning is mainly used to do fast adaptation with only few samples for users and items, including the item cold-start/warm-up problem~\citep{vartak2017meta,Pan:2019:WUC:3331184.3331268} and the user cold-start problem~\citep{dong2020mamo,lee2019melu,bharadhwaj2019meta}. For the item cold-start problem, linear weight adaptation is introduced to represent the items by the users who have viewed it~\citep{vartak2017meta} while Meta-Embedding is given to generate different embeddings for cold-start processes~\citep{Pan:2019:WUC:3331184.3331268}. For the user cold-start approaches, MAML is used to address the challenge by learning fast adaptation to new users. \rs{But those approaches do not consider the task overfitting problems. Differently, we are the \emph{first} to investigate the personalized learning rate based meta learning to address the cold-start user overfitting problem in RS, considering both major and minor users together.}

\section{Problem Formulation}

This section discusses the gradient-based meta learning in RS and how user imbalance affects the performance of MAML. Formally, following~\citep{lee2019melu,finn2017model}, the objective of meta learning in RS scenario can be described as:
$
    { \arg\min\nolimits_{\boldsymbol{\theta}} } \left[\frac{1}{|\tau|}\sum\nolimits_{\tau_i\in \tau} L_{\tau_i}\left(\boldsymbol{\theta}_{i}\right)\right] \text { s.t. } \boldsymbol{\theta}_{i}=\boldsymbol{\theta}-\alpha \nabla_{\boldsymbol{\theta}} L_{\tau_{i}}(\boldsymbol{\theta}),
$
 where $i$ is the $i$-th index for user $u_i \in \mathcal{U}$ (or task, since each user is treated as a task in our RS settings), $L_{\tau_{i}}$ is the loss function for user $i$ (e.g., mean squared error) and the item-feedback set $\tau_{i} = \{\langle \text{item}^j_i, \text{score}^j_i\rangle_{j=1}^{\hat{N}}\}$ is the set of items $\langle \text{item}^j_i\rangle_{j=1}^{\hat{N}}$ ( $\hat{N}$ items in total) and their corresponding feedbacks $\langle \text{score}^j_i\rangle_{j=1}^{\hat{N}}$  (e.g., ranking score) by user $u_i$. $L_{\tau}$ is the loss for different tasks and $\tau\supseteq \tau_i$ is the union set of 
the item-feedback subsets for all the users. %
 $\alpha$ is the (inner) Learning Rate (LR).  For simplification, we use $L_{\tau}$  to replace$\frac{1}{|\tau|}\sum\nolimits_{\tau_i\in \tau} L_{\tau_i}$. \rs{Before we discuss the user overfitting problem, we present a special case that helps us to analyze the problem better: users are clustered into two groups with densities $p_1$, $p_2$ ($p_1\geq p_2$). Users' actual preferences are $x_1$, $x_2$ for group $1$ and $2$ respectively. Now, we show that MAML ignores the minor users in this case:}
 \begin{lemma}\label{distribution}
 \rs{Suppose that the loss function is defined as  $L = \sum\nolimits_{i=1}^2p_i(\theta_i- x_i)^2$ , where $x_1,x_2$ are the actual preferences for task $1$ and $2$, and $\theta_i = \theta-\alpha \nabla_{\theta} L_{\tau_{i}}(\theta)$, when $p_1 \geq p_2$, $x_1\geq x_2$, we have  $(\theta^*_1- x_1)^2 \leq (\theta^*_2 - x_2)^2 $, where $\theta^*$ is the optimal parameter $\theta$.}
 \end{lemma}
\rs{All the proofs can be found in Appendix A. This case indicates that the MAML method tends to optimize the major users preferentially.}

\section{Methodology}

This section presents a practical adaptive Learning Rate (LR)-based meta-learning method to address the user overfitting problem, including  1) why adaptive LR can solve the user overfitting problem and how to apply adaptive LR by end-to-end training, 2) a similarity-based learning approach to improve the performance by considering similar users and an approximated tree-based implementation for fast search, and 3) a memory agnostic regularizer to further reduce space complexity to constant.

\subsection{Adaptive Learning Rate based MAML}

To consider both major and minor users together, we propose an adaptive learning rate approach to design different gradient steps for different users :
\begin{equation}\label{maml}
\small
        {\arg \min\nolimits_{\boldsymbol{\theta}}} \left[\frac{1}{|\tau|}\sum\nolimits_{\tau_i\in \tau} L_{\tau_i}\left(\boldsymbol{\theta}_{i}\right)\right],
        \text { s.t. } \boldsymbol{\theta}_{i}=\boldsymbol{\theta}-\alpha(h_i) \nabla_{\boldsymbol{\theta}} L_{\tau_{i}}(\boldsymbol{\theta}),
\end{equation}
$\!$where $h_i\in \mathbb{H}$ is the $i$-th user's feature embeddings\footnote{\rs{User's feature embeddings can include different user information like the user profile embeddings and user browsing history embeddings.}} (or user embedding for short). The difference between ours and MAML is that the LR in our method is a mapping from the user embedding to a real number rather than a fixed LR. Intuitively, with an adaptive learning rate, the meta agent can fit any user even if it is far from the meta strategy (like the user 4 in Fig.~\ref{meta_reason}). \rs{Here, an analysis is given to illustrate the adaptive learning rate can get better results in user imbalanced dataset:}

\begin{lemma}\label{bound}
Based on Lemma \ref{distribution}, we further defined as $L^\prime =\sum\nolimits_{i=1}^2p_i(\theta^\prime_i- x_i)^2$, $\text { s.t. } {\theta}^\prime_{i}={\theta}-\alpha(h_i) \nabla_{{\theta}} L_{\tau_{i}}({\theta})$, ($x_1,x_2$ are the target values for task $1$ and $2$, where $x_2> x_1$) there exists $\alpha_1$ and $\alpha_2$ that satisfying
\begin{equation*}
(\theta^{*}_2- x_2)^2 \geq (\theta^{*\prime}_2 - x_2)^2 \text{and}\  L^*\geq L^{*\prime},
\end{equation*}
where $L^*$ and $L^{*\prime}$ are the optimal loss functions for standard MAML and adaptive MAML respectively. $\theta^{*\prime}$ is the optimal $\theta^{\prime}$ parameters.
\end{lemma}

 \rs{Lemma \ref{bound} indicates that learning a personalized adaptive learning rate provides a lower loss value for minor users than the standard MAML method, and is able to achieve a lower total loss value than standard method ($L^*\geq L^{*\prime}$). This coincides with our conjecture that the adaptive LR meta learning methods perform better than the fixed LR methods when facing an imbalanced distribution.}

\rs{Therefore, the main challenge is to capture the relationship among the user distribution,  the parameters of network, and user own learning rate.} To construct it, we leverage the end-to-end data-driven learning approach to map each user's features into different high-dimensional embedding and obtain the individual learning rate adaptively. Based on this idea, we propose the Personalized Adaptive Meta Learning (PAML) method, a direct way to train the network by taking the gradient descent of Eq. (\ref{maml}):
\begin{align*}
       & \boldsymbol{\theta} = \boldsymbol{\theta} - \beta \sum\nolimits_{\tau_i}\nabla_{\boldsymbol{\theta}}L_{\tau_i}(\boldsymbol{\theta}_i)\Bigg( I - \alpha(h_i;\boldsymbol{\psi})\nabla^2_{\boldsymbol{\theta}} L_{\tau_{i}}(\boldsymbol{\theta})\Bigg), \\ &
   \boldsymbol{\psi} = \boldsymbol{\psi}  + \beta\nabla_{\boldsymbol{\psi}} L_{\tau}(\boldsymbol{\theta}) \left(\sum\nolimits_{\tau_i}\nabla_{\boldsymbol{\psi                 }}\alpha(h_i;\boldsymbol{\psi})L_{\tau_{i}}(\boldsymbol{\theta})\right) .
\end{align*}
 where $\boldsymbol{\theta}$ and $\boldsymbol{\psi}$ are the parameters for the model and learning rate\footnote{For simplification, we use $\alpha(h_i)$ to replace  $\alpha(h_i;\boldsymbol{\psi})$ if not specially mentioned.}. $\beta$ is the outer learning rate. The two equations above are drawn by the chain rule. %

\subsection{Approximated Tree-based PAML}\label{trees}

The key for PAML is to get accurate personalized $\alpha(h)$ for different users. However, directly using Fully-Connected (FC) layers to learn $\alpha(h_i)$ is hard since LRs are related to the task distribution and FC layers are not capable to memorize a large number of the users. To address the challenge, one straightforward idea is to consider other similar users' features as a reference since users with similar feature embeddings share similar LRs. Based on this idea, we introduce the similarity-based method to find users with similar high-level features (the embedding) and interests.

Formally speaking, when a new user $u_i$ with embedding $h_i\in \mathbb{H}$ comes, our goal is to find the users with the most similar embeddings to $u_i$. Here, we define a similarity function $s_k:\mathbb{R}^m\times \mathbb{R}^m \rightarrow \mathbb{R}$ to estimate the embedding similarity between user $u_i$ and user $u_k$.  After computing top-$K$ nearest (the most similar) users with the values of similarity, these values can be treated as a reference to obtain a personalized adaptive learning rate. Thus, we can get an LR function for a new user $u_i$: $\alpha(h_i) := \alpha^{\prime}(h_i) + \tilde{\alpha}$, by considering the user information and the user similarities together, where  $\tilde{\alpha} = \sum_k \frac{s_k}{\sum_j s_j + \sigma} \alpha(h_k)$ is the weighted average sum of the nearest users' LRs, $s_k$ is the $k$-th existing nearest user's similarity value to user $i$, $\alpha^{\prime}(h_i) $ is neural network modules mapping from user embedding to a real value, and $\sigma$ is a small value ($10^{-5}$). Here, we set the Gaussian kernel function $s_k=\exp \left(-\delta\left\|h_i-h_k\right\|^{2}\right)$ as the similarity function, where $\delta$ is a constant and $\|\cdot\|$ is the 2-norm. Due to the limited space, details about how the kernel-based function can be used to estimate the similarity can be found in Appendix B. 

Since we need to find the top-$K$ nearest users in every training step, it is necessary to find a fast searching approach to accelerate the training process. %
Here, we leverage the kd-tree~\citep{muja2014scalable} as the basic structure to store users' embeddings. Specifically,  we first initialize the tree structure with several users' ($\geq K$) embeddings and its corresponding LRs (the warm-up stage). Then, when a new user $u_i$ with embedding $h_i$ comes, we search the $K$-th nearest users in the tree. After that, we add the new user embedding and its LR into that tree and rebuild that tree. We remove the least frequent used user's embedding and its LR when the tree-size is larger than the memory size we set manually. Based on this structure, the time complexity of searching each user can be reduced from $\mathcal{O} (n)$ (brute force method) to $\mathcal{O} (\sqrt{n})$ for two-dimension embedding cases~\citep{yianilos1993data}. %

However, this structure induces bias because the user embedding layers are dynamically updated during the training process but the previous users' embeddings stored in the tree nodes are fixed, causing the new embeddings and the old embeddings unfitted. In order to coordinate them, we also let the tree nodes be dynamically updated by gradient descent: $\text{node}_{j} = \text{node}_{j} - \beta\nabla_{\text{node}_{j}} L_{\tau}$, where $\text{node}_{j}$ indicates the $j$-th node value (the user's embedding).  The tree-based structure will also be updated once the embeddings change. The learning-rates stored in nodes are also updated similarly.

To further reduce space usage and speed up, we use the approximated store and search methods in place of the exact (precious) ones. That is, we can take advantage of the randomized kd-tree algorithm~\citep{muja2014scalable} to achieve the approximation. These structures can be easily implemented by the open-source python package (e.g., Pyflann~\citep{muja2013flann}). We name it as Approximated-Tree PAML (AT-PAML). The loss function is the same as Eq.~\eqref{maml} and the pseudo-code can be found in Algs. \ref{alg1} and \ref{alg2}.

\begin{algorithm}[t]
   \caption{Personalized Adaptive Meta Learning}\label{alg1}
\KwInput{User distribution $p(\mathcal{U})$, the learning rate $\beta$.}
 Initialize the meta-policy with parameters $\theta$\;
 \For{episode $T$}{
   sample $N$ of users from $p(\mathcal{U})$\;
  \For{each user $i \in U$}{Split the support set $\tau_i$ and query set $\hat{\tau}_i$ randomly\;
  Extract user embedding $h_i$\;
  \uIf{\emph{Approximated Tree-based method}}{
     \uIf{Warm-up stage}{
    Set LR as $\alpha$\;
    Store$\_$Node($h_i$, $\alpha$)\;
  }
  \Else{ $\langle h_k \rangle^K_{k=1}, \langle \alpha(h_k) \rangle^K_{k=1}$ = Search$\_$Tree($h_i$)\;
     $s :=\{s_k\}^K_{k=1}= \{\langle \phi(h_i),  \phi(h_k) \rangle^{2}_{\mathcal{H}}\}^K_{k=1}$\;
      Obtain $\alpha(h_i) = \sum_k \frac{s_k}{\sum_j s_j + \sigma} \alpha(h_k) + \alpha^{\prime}(h_i)$\;
    Store$\_$Node($h_i$, $\alpha(h_i)$)\;}
  } \ElseIf{Regularized-based method}{Obtain the learning rate $\alpha$ directly by $\alpha(h_i)$;}
    Update 
   $\theta^{i} \leftarrow \theta -  \alpha(h_i)\nabla_\theta \mathcal{L}_{\tau_{i}}(;\theta)$ with support set $\tau_i$\;
  }
   
   \uIf{Regularized-based method}{update 
   $\theta$ by Eq. (4) with query set $\bm\hat{\tau}= \langle\hat{\tau}_i\rangle_{i=1}^N$\;}   
   \ElseIf{\emph{Approximated Tree-based method}}{update 
   $\theta$ by Eqs.~(5) $\&$ (6) with query set $\bm\hat{\tau}= \langle\hat{\tau}_i\rangle_{i=1}^N$. (Eqs. (4), (5) and (6) are shown in Appendix C).}   }
 \KwOutput{a well-trained meta-strategy.}
\end{algorithm}

\begin{algorithm}
  \caption{Store and Search approaches }\label{alg2}
  \SetKwProg{Fn}{Def}{:}{}
    \SetKwFunction{FS}{Store$\_$Node}
\Fn{\FS{$h_i$, $\alpha$}}{  \If{Memory is full}{
  Remove the embeddings and LR that are least recently used\;
  }
    Store the embeddings, LR, and build the standard kd-tree by embedding~\citep{friedman1976algorithm}\;
  \KwRet the kd-tree\;
  }
  
  \SetKwFunction{FSub}{Search$\_$Tree}
  \Fn{\FSub{$h_i$}}{
Search the nearest neighbours~\citep{yianilos1993data} and their corresponding LRs\;
\KwRet  the nearest neighbours and their corresponding LRs \;
  }
\end{algorithm}

\begin{figure*}%
  \centering
  \centerline{		\includegraphics[width=0.75\linewidth]{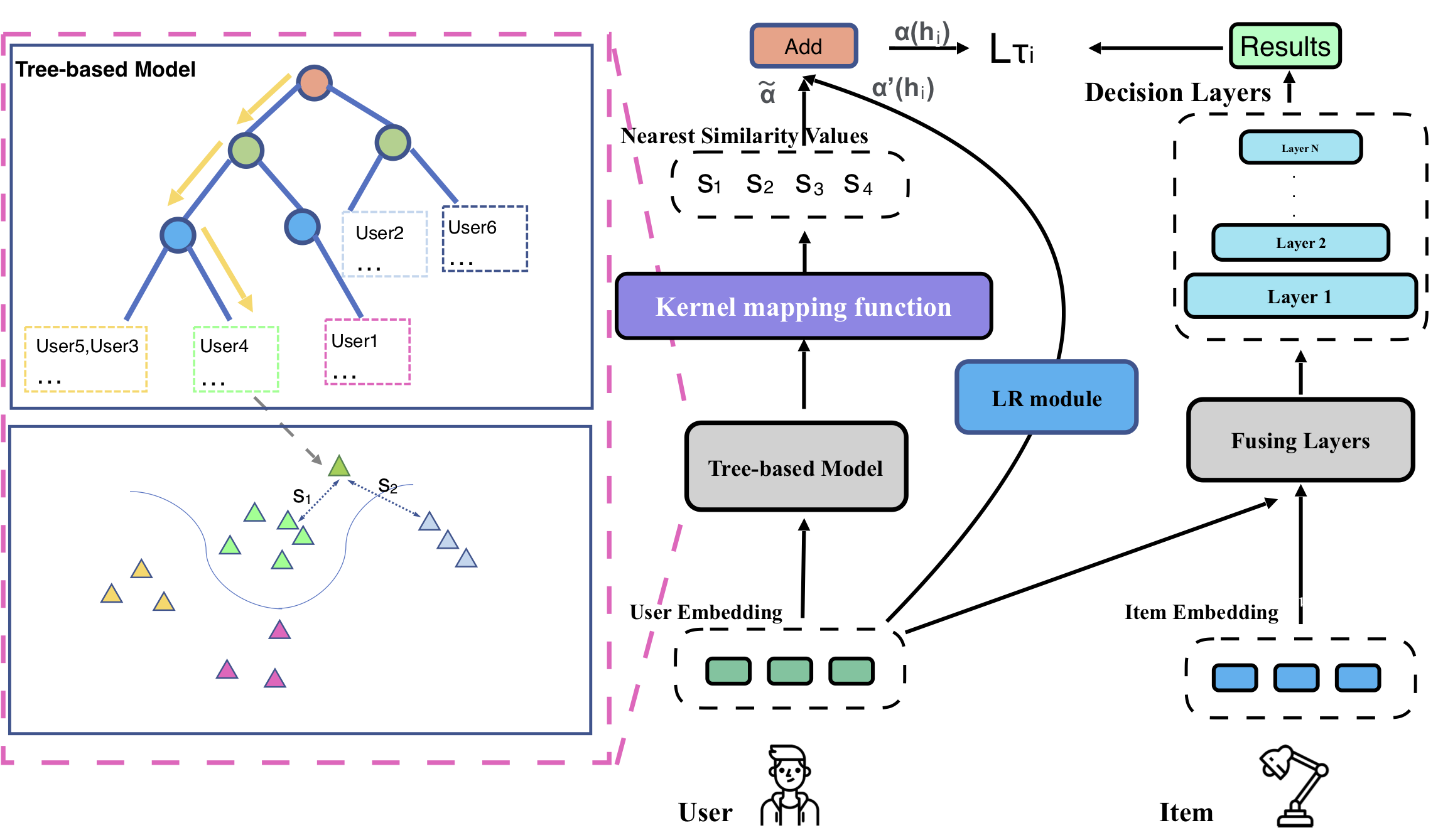}}
\caption{The network structure of the tree-based PAML. The right part is the network for inner update while the left part is the tree-based module. Specifically, \emph{Fusing Layer} indicates the concatenation of user and item embeddings, \emph{Add} is an adding function: $\alpha^{\prime}(h_i) + \tilde{\alpha}$. $Layer 1, \dots N$ means fully connected layers and $S_1, \dots S_4$ are similarity values. } \label{sys}
\end{figure*}

\subsection{Regularizer-based PAML }\label{mar}
However, for large-scale real-world applications (e.g., Netflix) which involve more than millions of users, even the linear space complexity methods (the space size is proportional to the number of user embeddings) to store user embedding is unacceptable. Thus, a constant space complexity algorithm needs to be proposed. %
A widely adopted method to achieve long-term memorization is the continual learning approach~\citep{aljundi2018memory,kirkpatrick2017overcoming}, which designs losses and training schemes to memorize the previous tasks. Inspired by their works, we design an auxiliary loss function to help the model to remember the users it has seen. That is, for any two users $u_i$ and $u_j$, memorizing the information of those two users means the loss values for those users are similar: i.e., $  \| L_{\tau_i}\left(\boldsymbol{\theta}_{i}\right) -     L_{\tau_j}\left(\boldsymbol{\theta}_{j}\right)\|$. Similarly, the multi-user loss is $\sum\nolimits_{u_i,u_j\in\mathcal{U};i>j}\|L_{\tau_i}\left(\boldsymbol{\theta}_{i}\right) -     L_{\tau_j}\left(\boldsymbol{\theta}_{j}\right)\|$. It is easy to understand because forgetting a task means its loss value for that task is higher than the loss value of other tasks. If the model can obtain qualified and similar loss values for each user, we can guarantee that the model has memorized all existing users' information. 

However, it is hard to implement it into practice since computing multi-user loss needs to calculate all the user-pair losses together. To address this challenge, we turn to optimize its upper bound:
\begin{small}
\begin{align*}
      \sum\nolimits_{u_i,u_j\in\mathcal{U};i>j}\|L_{\tau_i}(\boldsymbol{\theta}_{i} ) -     L_{\tau_j}\left(\boldsymbol{\theta}_{j}\right)\| \leq \\ \sum\nolimits_i  (|\mathcal{U}|-1)||\frac{\partial L_{\tau_i}}{\partial\boldsymbol{\theta}}||^2|| \alpha(h_i)|| + C,
\end{align*}
\end{small}
\noindent where $C$ is a constant. \emph{The proof of the upper bound can be found in Appendix~A.2}. Notice that the sum of $||\frac{\partial L_{\tau_i}}{\partial\boldsymbol{\theta}}||^2|| \alpha(h_i)||$ can be divided into each task (user) loss as a regularizer to enable the network to remember each user information: i.e., for user $i$, the regularizer $L^r_i$ is set as:
$L^r_i = ||\frac{\partial L_{\tau_i}}{\partial\boldsymbol{\theta}}||^2|| \alpha(h_i)||$. Intuitively, when $L^r_i = 0$, it indicates that the meta strategy has already obtained a good score for user $i$ without training. Therefore, this term reinforces the model to remember the users it meets. Moreover, since $|\mathcal{U}|-1$ is a constant, we can replace it with a positive real number $\gamma$  to balance $L^r_{\tau_i}$ and $L_{\tau_i}$. We name the PAML with this REGularizer as \emph{REG-PAML}.  Since we do not need any extra space to store user embedding, the space complexity is \emph{constant}. The total loss function for REG-PAML is $   \mathcal{L} = \sum\nolimits_i L_{\tau_i} +  \gamma L^r_i$. The relationship to implicit MAML~\citep{rajeswaran2019meta} can be found in Appendix~C. The pseudo-code can also be found in Alg.~\ref{alg1}.

\section{Empirical Studies}\label{exps}
This section validates our methods in various environments and tasks, including the rating prediction (in the MovieLens-1M dataset and the BookCrossing dataset) as well as the CTR prediction (the real-world production dataset). We also analyze the performance of our methods for both the minor and major users and conduct ablation studies. %

\subsection{Experimental Setup}

\textbf{Datasets.} We use both open-source datasets as well as real-world production dataset to evaluate the performance of our methods in user imbalanced dataset, including MovieLens-1M~\citep{harper2015movielens}, BookCrossing~\citep{ziegler2005improving} as well as production dataset (collected from Taobao e-commerce platforms, which is somewhat similar to~\citep{zhao2018impression,guo2019securing}). The data pre-processing scheme is:  1) We rank the users with the number of their log data and choose the $80\%$ users least log data as the cold-start users (the value $80\%$ is based on the Pareto Principle). 2) We randomly split the users by $7:1:2$ for training, validation, and testing. 3) We tick out the users when they either have blank or wrong features (including the ages are less than 10 or larger than 100 as well as the location features with garble) or the items they viewed are less than a threshold (two items for rating prediction and two clicked items for CTR prediction). 4) We separate the support and query sets for each user with a ratio of 80$\%$:20$\%$ randomly. Details about the features we use, and the statistics can be found in  Appendix D.1 $\&$ D.3. %

\rs{\textbf{Major and minor users separating.} To testify whether our methods provide better prediction results for both minor users and major users, we need to define a criteria to split the major and minor users. However, as we mentioned above, finding the minor users are hard because there is not an exact method to precisely cluster the users by their features. Here we define a simple but not completely precise rule to \emph{approximately} split the minor and major users: since distribution of the feature embeddings is the key to classify whether a user is minor or major user, we split the minor and major users by the following criteria: if the user with more than two of features which are in the top $30\%$ largest number of feature values set, we regard it as the major user. Otherwise, it is the minor user. The major and minor users are $65\%$ and $35\%$ in MovieLens, $72\%$ and $28\%$ in BookCrossing, and $71\%$ and $29\%$ in the Taobao dataset. }

\textbf{Baselines.} To validate our methods, we compare them with several State-Of-The-Art (SOTA) methods, including MeLU~\citep{lee2019melu}, Meta-SGD~\citep{li2017meta}, and transfer learning~\citep{tan2018survey}. 1) MeLU~\citep{lee2019melu}. MeLU is one o f the state-of-the-art gradient-based meta-learning methods in RS. Since it is very similar to other user-based meta learning~\citep{finn2017model,bharadhwaj2019meta}, we use it as the standard gradient-based meta-learning method. 2) Meta-SGD~\citep{li2017meta}. Meta-SGD is an adaptive learning rate meta-learning method. Comparing with standard MAML using a fixed (inner) learning rate, it uses different learning rates for each parameter. But those adaptive parameters do not rely on the user embedding. We use it as a baseline to testify the performance of our embedding-based adaptive LR methods. 3) Transfer-learning~\citep{tan2018survey}. We choose the standard transfer-learning approach which has been widely adopted to address the user imbalance issue as the baseline. Specifically, we train the model with all users in training data and fine-tune the trained model in each user support set in the test data.

\textbf{Network structures.} For the AT-PAML, as shown in Fig.~\ref{sys}, it contains five modules, the embedding modules (including the embedding layers for each feature respectively, we use an FC layer to concatenate them together), the tree-based module, the learning-rate module, the decision module as well as the fusing module. The embedding modules are the standard embedding layers~\citep{pennington2014glove}. The tree-based module is what we discussed in Sec.~\ref{trees}. The LR module is two FC layers with ReLU as activation function. The decision module (except for the last layer) is based on the FC layer with ReLU as activation function. For the CTR prediction, the last layer is a two-dimension softmax layer. For the rating prediction, it is a real value without any activation function. The fusing layer is to concatenate all user and item embeddings together. More details can be found in Appendix~D.2. For the REG-PAML and PAML, except that they do not have the tree module, other modules are similar. 

\textbf{Our methods.} We validate our proposed methods, including the \emph{Approximated Tree-based PAML}(AT-PAML) as well as \emph{Regularizer PAML} (REG-PAML). For ablations, we implement vanilla PAML and REG-PAML with different $\gamma$ values (we use $\gamma=10^{-3}$ as default gamma value). %
We use Pytorch\footnote{\url{https://pytorch.org/}} as the deep learning library and we use the ADAM as the optimizer~\citep{kingma2014adam}. All the experiments are done on a single GeForce GTX 1080 Ti. The production dataset is obtained by the company's recommendation pipeline. The package for the two-tailed student's t-test is from\footnote{\url{https://docs.scipy.org/doc/scipy/reference/generated/scipy.stats.ttest_ind.html}}. he running time for REG-PAML and PAML are about 2 hours, and for AT-PAML is about 4 hours on a single GTX 1080 GPU for MovieLens. For memory space, the space size for AT-PAML is 2.12GB, while that of the REG-PAML/ PAML is about 669MB in MovieLens.  %

\textbf{Evaluation metrics.} Following~\citep{lee2019melu,ren2019lifelong}, we leverage different evaluation metrics to testify different aspects of our methods, including Mean Squared Error (MSE) and average Normalized Discounted Cumulative Gain (NDCG) for rating prediction results; as well as Area Under the Curve (AUC) and weighted Negative Entropy Loss (NEL) for CTR prediction results. \rs{To evaluate the performance difference between the minor and major users, we use the two-tailed student's t-test for statistical hypothesis test. The null hypothesis is that 2 independent samples (minor users and major users) have identical average (expected) values. If the p-value is high, indicating that the expected values of the minor users and major users are not different. On the contrary, when the p-value is low, it shows that the expected values of the minor users and major users are different. Therefore, when a method focuses on both minor and major users, it does \textbf{not} have a low p-value.} More detailed definitions can be found in Appendix~D.6. 

\begin{table*}[t]
    \centering
    \small
    \captionsetup{justification=centering}
{
\begin{tabular}{c|c|c|c|c|c|c}

\toprule
\multirow{2}{*}{} & \multicolumn{3}{c|}{\textbf{MovieLens-1M}} & \multicolumn{3}{c}{\textbf{BookCrossing}} \\ \cline{2-7}
 & \textbf{ MSE $\downarrow$} & \textbf{Avg nDCG@3 $\uparrow$} & \textbf{Avg nDCG@5 $\uparrow$} &\textbf{ MSE $\downarrow$} & \textbf{Avg nDCG@3 $\uparrow$} & \textbf{Avg nDCG@5 $\uparrow$} \\ \midrule
\scriptsize{MeLU}~\citep{lee2019melu} & 1.451\scriptsize{$\pm$0.022} & \textbf{0.793\scriptsize{$\pm$ 0.002}} & 0.800\scriptsize{$\pm$ 0.001} &\emph{4.019} \scriptsize{$\pm$0.101} &0.864\scriptsize{$\pm$0.002} &  0.914\scriptsize{$\pm$0.002} \\ 
\scriptsize{Meta-SGD}~\citep{li2017meta} & 1.340\scriptsize{$\pm$ 0.054} & 0.773\scriptsize{$\pm$ 0.007} & 0.794\scriptsize{$\pm$ 0.007} & 5.197\scriptsize{$\pm$ 0.089} & \emph{0.867\scriptsize{$\pm$ 0.001}} &\emph{0.921\scriptsize{$\pm$ 0.002}}\\ 
\scriptsize{Transfer-Learning}~\citep{tan2018survey} & \emph{1.308\scriptsize{$\pm$0.016}} & 0.778\scriptsize{$\pm$ 0.001}& 0.796\scriptsize{$\pm$ 0.002} & 4.522\scriptsize{$\pm$0.376}  & 0.862\scriptsize{$\pm$0.010} &  0.919\scriptsize{$\pm$0.005}\\
\hline
\scriptsize{\textbf{AT-PAML (ours)}} & 1.322\scriptsize{$\pm$0.007} &  \emph{0.788\scriptsize{$\pm$0.003}} & \textbf{0.806$^*$\scriptsize{$\pm$0.006}} & \emph{3.991\scriptsize{$\pm$0.114}}  &0.852\scriptsize{$\pm$0.001} &  \textbf{0.930$^*$\scriptsize{$\pm$0.002}}\\
\scriptsize{\textbf{REG-PAML (ours)}} & \textbf{1.210$^*$\scriptsize{$\pm$ 0.029}} & 0.779\scriptsize{$\pm$ 0.002}& \emph{0.803\scriptsize{$\pm$ 0.001}} &\textbf{3.928$^*$\scriptsize{$\pm$ 0.176}}
& \textbf{0.868\scriptsize{$\pm$ 0.001}} &0.917$\pm$ 0.001\\
\bottomrule
\end{tabular}
}    \caption{%
    Comparison of different methods on the MovieLens and the BookCrossing datasets. The best results are highlighted in \textbf{bold} and the second-best results are in \emph{italic}. \textbf{Avg} means average. The mean and standard deviation are reported by 3 independent trials. $*$ denotes statistically significant improvement over the best baseline method (measured by t-test with p-value$<0.05$).
    }
    \label{rat_table}
\end{table*}

\begin{table}
    \centering
    \scriptsize
    \captionsetup{justification=centering}
{
\begin{tabular}{|c|c|c|c|}
\toprule
\multirow{2}{*}{} & \multicolumn{3}{c|}{\textbf{MovieLens}} \\ \cline{2-4}
 & \textbf{MSE} & \textbf{Avg nDCG@3} & \textbf{Avg nDCG@5}  \\ \midrule
 \textbf{REG-PAML ($\gamma=10^{-5}$)} & {1.509\scriptsize{$\pm$ 0.029}} & 0.773\scriptsize{$\pm$ 0.004}& {0.799\scriptsize{$\pm$ 0.003}} \\
\textbf{REG-PAML ($\gamma=10^{-4}$) } & {1.227\scriptsize{$\pm$ 0.026}} & 0.779\scriptsize{$\pm$ 0.003}& \emph{0.803\scriptsize{$\pm$ 0.002}} \\
\textbf{REG-PAML ($\gamma=10^{-3}$)} & \emph{1.210\scriptsize{$\pm$ 0.013}} & \textbf{0.791*\scriptsize{$\pm$ 0.002}}& \textbf{0.814*\scriptsize{$\pm$ 0.004}} \\
\textbf{REG-PAML ($\gamma=10^{-2}$)} & \textbf{1.191*\scriptsize{$\pm$ 0.021}} & 0.777\scriptsize{$\pm$ 0.002}& {0.803\scriptsize{$\pm$ 0.004}} \\
\textbf{PAML} & {1.520\scriptsize{$\pm$ 0.034}} & 0.787\scriptsize{$\pm$ 0.001}& \emph{0.809\scriptsize{$\pm$ 0.004}} \\
\textbf{AT-PAML } & 1.322\scriptsize{$\pm$0.007} &  \emph{0.788\scriptsize{$\pm$0.003}} & {0.806\scriptsize{$\pm$0.006}} \\
\bottomrule
\end{tabular}
}
\caption{%
    Ablations of different methods on MovieLens-1M. The best results are highlighted in \textbf{bold} and the second-best results are in \emph{italic}. \textbf{Avg} means average. The results are reported by 3 independent trials. $*$ denotes statistically significant improvement over PAML (measured by t-test with p-value$<$0.05).
    }
        \label{abs_three}
\end{table}

\begin{figure*}[tb!]
\subfigure[The AUC for different methods in test dataset. TL is the abbreviation for transfer learning. Higher is better.]{
\includegraphics[width=0.3\textwidth]{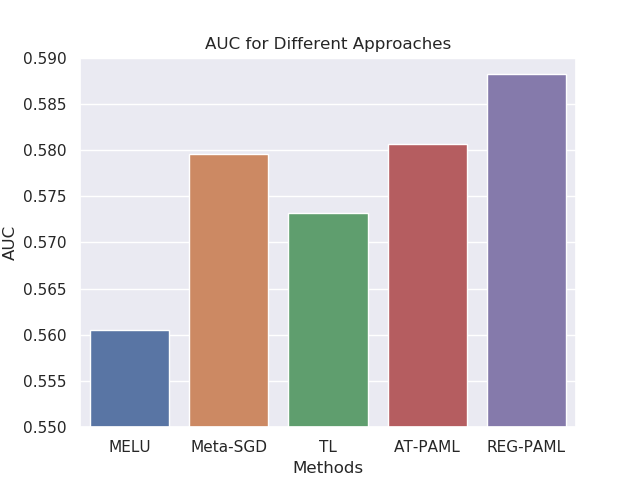}
\label{tb_auc}}\hfill
\subfigure[The box plot of the user NEL in test dataset.]{
\includegraphics[width=0.3\textwidth]{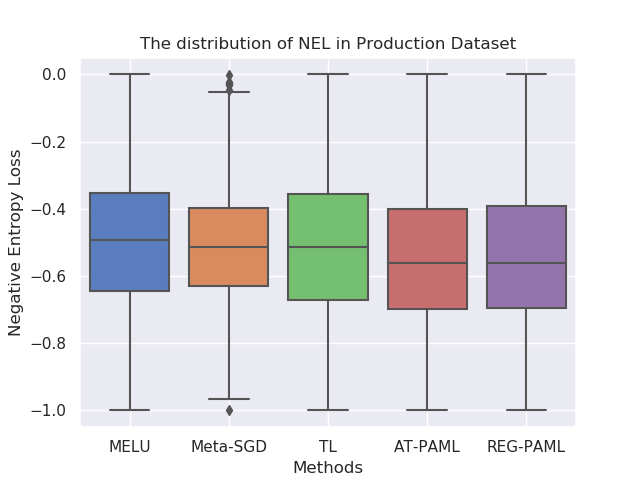}
\label{tb_all}}\hfill
\subfigure[The NEL for different methods in the production dataset. Lower is better.]{
\includegraphics[width=0.3\textwidth]{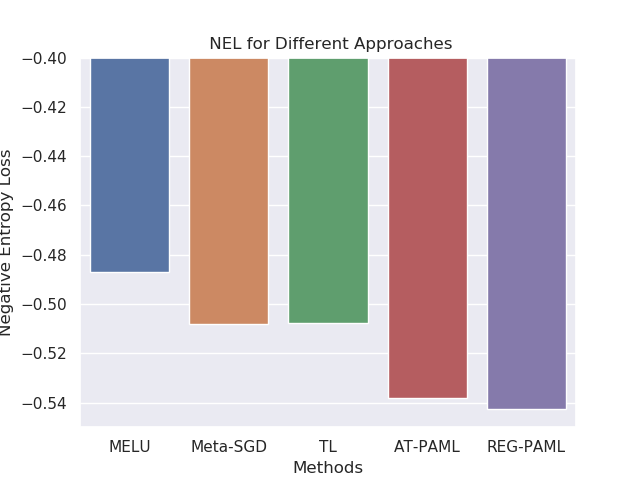}
\label{tb_nel}}\hfill
\caption{The performance and statistics of different methods in the production dataset. The results are reported by 3 independent trials. The NEL for REG-PAML has statistically significant improvement over the best baseline method (measured by t-test with p-value$<0.05$).}
\label{fig:matrix_train}
\end{figure*}

\begin{figure*}[tb!]
\centering
\subfigure[The distribution of minor and major  user LRs for REG-PAML ($\gamma = 10^{-3}$).]{
\includegraphics[width=0.3\textwidth]{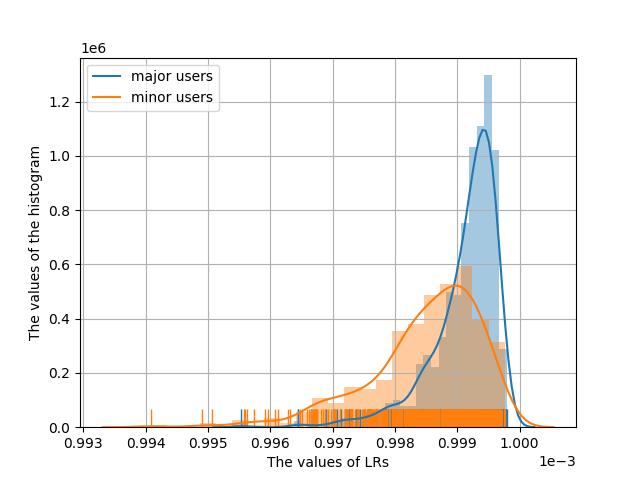}}\hfill
\subfigure[The distribution of minor and major user LRs for REG-PAML ($\gamma = 10^{-5}$).]{
\includegraphics[width=0.3\textwidth]{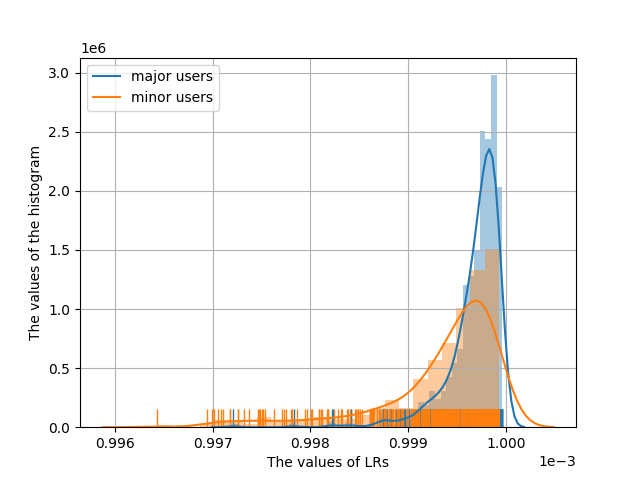}}\hfill 
\subfigure[The distribution of minor and major  user LRs for PAML.]{
\includegraphics[width=0.3\textwidth]{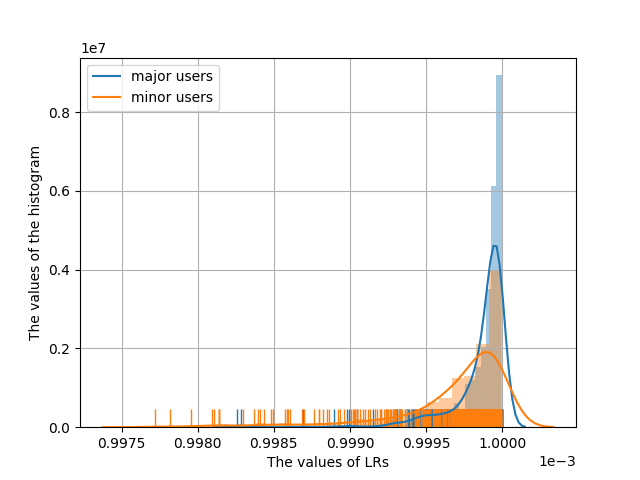}}\hfill
\caption{The distribution of LRs for REG-PAML ($\gamma=10^{-3}$), REG-PAML ($\gamma=10^{-5}$), and PAML. }\label{lrs}
\end{figure*}

\begin{table}[t]
    \centering
    \scriptsize   
    \captionsetup{justification=centering}
{
\begin{tabular}{|c|c|c|c|}

\toprule
\multirow{2}{*}{} & \multicolumn{3}{c|}{\textbf{MovieLens-1M}} \\ \cline{2-4}
  & \textbf{MSE (MAU)} & \textbf{MSE (MU)} & \textbf{P-value } \\ \midrule 
 \textbf{MELU}  & 1.413\scriptsize{$\pm$ 0.025} & 1.506\scriptsize{$\pm$ 0.021} & 0.071\\
\textbf{Meta-SGD}  &1.300\scriptsize{$\pm$ 0.059}&  1.403\scriptsize{$\pm$ 0.045}  &0.024\\
\textbf{Transfer-Learning}  & 1.288\scriptsize{$\pm$ 0.016}& \emph{1.340\scriptsize{$\pm$ 0.013}}  & 0.227\\
 \textbf{REG-PAML (ours) }  &\textbf{1.194}\scriptsize{$\pm$ 0.019}&\textbf{1.236\scriptsize{$\pm$ 0.011} } &0.312\\
 \textbf{AT-PAML (ours)} &   \emph{1.284\scriptsize{$\pm$ 0.005}}& 1.382\scriptsize{$\pm$ 0.011}  & 0.018\\
\bottomrule
\end{tabular}
}   \caption{%
    Results of different MSEs for MAjor Users (MAU) and Minor Users (MU) on MovieLens. The best results are highlighted in \textbf{bold} and the second-best results are in \emph{italic}. P-value is drawn by two-tailed student’s t-test between the minor user MSEs and major user MSEs over the same methods. 
    }
\label{ls_methods}
\end{table}

\subsection{Results and Discussions}
\textbf{Rating Prediction.} We first analyse the rating prediction results. As shown in Table~\ref{rat_table}, our methods (AT-PAML $\&$ REG-PAML) outperform other SOTA methods in different metrics, especially for REG-PAML, which has reached 3 out of 6 best results. For the AT-PAML, it has two best results and two second-best results, implying that this method is also well-performing. From the Figs.~5 and~6 in Appendix~D.7, we can find that REG-PAML has fewer outliners, indicating the robustness of REG-PAML. Moreover, in Fig.~5, Meta-SGD has some outliners with large values, which might be the reason that  Meta-SGD does not perform well in the BookCrossing dataset. Moreover, both the AT-PAML and REG-PAML have lower third quartile comparing with SOTAs, showing that our methods can achieve good-enough performance for most of the cold-start users.

\textbf{CTR prediction.} We then testify our method in the CTR prediction task. As shown in Figs.~\ref{tb_auc} and ~\ref{tb_nel},  Our methods achieve the best (0.5882) and the second-best (0.5806) results in the AUC metric comparing with other SOTA methods (the best result of SOTA methods is 0.5796). One interesting phenomenon is that for AUC, MeLU gets the worst result among all the methods, indicating that the vanilla meta learning does not work well in this dataset. For the NEL, our methods also get the top two results (-0.5382,-0.5426) among all other methods (the best result of SOTA methods is -0.5080). %
Moreover, MeLU does not perform well enough comparing with ours, indicating the importance of building a learnable mapping between user embedding to LR. Fig.~\ref{tb_all} reveals some statistical information on the NEL. We can find that the mean and the first quartile of the NEL distribution for Meta-SGD are larger than ours. This result reveals that our methods can make better prediction results than other meta learning methods for both the minor and major users. %

\textbf{Minor users.} We conduct experiments to show whether our methods help to improve the performance of minor users. As shown in Tab.~\ref{ls_methods}, AT-PAML outperforms all the SOTAs in both minor users and major users significantly. Also, from the two-tailed student’s t-test analysis, we find that the AT-PAML has the highest p-value, indicating that with a high probability the expected values for minor users and major users are similar. We also notice that SOTA meta-learning methods (MeLU and Meta-SGD) have relatively low p-values, revealing that the expected values for minor and major users are different with high probability. The results reveal that our methods are not only well-performing in the whole datasets but also focus on both the minor and major users than the SOTA meta-learning methods. Moreover, we also visualize the LRs for each user, as shown in Fig.~\ref{lrs}. Combining with Tab.~\ref{abs_three}, the closer the LRs of minor and major users are, the higher the MSEs, showing different LRs do play an important role in different user preference predictions. One interesting phenomenon is that during the experiments, we notice a \emph{failure case} shown in Fig.~8 in Appendix~D.7, which reveals  when the method is not capable to distinguish the minor users and the major users, it cannot perform well. These results also empirically 
show the significance of different LRs for different users. Also, the t-SNE  ~\citep{maaten2008visualizing} visualization of user embeddings ( Fig.~11 in Appendix~D.7) show that our methods can distinguish the minor and major users. Tab.~\ref{abs_gamma} reveals that $\gamma$ affects the performance of both major and minor users.

\textbf{Ablations.} We do ablation studies to evaluate whether our proposed modules work. As shown in Tab.~\ref{abs_three}, PAML with tree-based structures and regularizer ($\gamma=10^{-3}$) perform better than the vanilla one under the MSE and the Avg nDCG@3 metrics. Also, AT-PAML and REG-PAML achieve lower values in the MSE metric for both minor and major users (Tab.~\ref{abs_gamma} in Appendix~D.7). An interesting phenomenon is that with the decrease of parameter $\gamma$, the MSE values become higher while the p-values become lower, meaning that the model performs worse and focuses more on major users. This result reveals that the regularizer does play an important role to achieve better performance and drive the model to focus on both minor and major users.

\section{Conclusion}

    In this paper, we propose a novel personalized adaptive meta-learning method to address the user overfitting problem in cold-start user preference prediction challenge with three key contributions: 1) We are the first to introduce a personalized adaptive learning rate based meta learning approach to improve the performance of MAML by focusing on both the major and minor users. 2)  We build an approximated tree-based method to store different users' features for fast calculation and a similarity-based method to find similar users as a reference to obtain accurate personalized learning rates for different users. 3) To reduce the memory usage, we propose a memory agnostic regularizer to reduce the space complexity and maintain the memorizing ability by learning. Experiments on MovieLens-1M, BookCrossing, and real-world production dataset reveal that our method outperforms the state-of-the-art methods dramatically for both the minor users and the major users.

\begin{table}[t]
    \centering
    \scriptsize
    \captionsetup{justification=centering}
{
\begin{tabular}{|c|c|c|c|}

\toprule
\multirow{2}{*}{} & \multicolumn{3}{c|}{\textbf{MovieLens-1M}} \\ \cline{2-4}
 & \textbf{MSE (MAU)} & \textbf{MSE (MU)} & \textbf{P-value} \\ \midrule 
 \textbf{REG-PAML ($\gamma=10^{-5}$)} & 1.464\scriptsize{$\pm$ 0.032} & 1.601\scriptsize{$\pm$ 0.021} & 0.008\\
\textbf{REG-PAML ($\gamma=10^{-4}$) } &  1.210\scriptsize{$\pm$ 0.026} & 1.255\scriptsize{$\pm$ 0.024} &0.226 \\
 \textbf{REG-PAML ($\gamma=10^{-3}$)} & \emph{1.196\scriptsize{$\pm$ 0.019}} &\textbf{1.183\scriptsize{$\pm$ 0.024}} & 0.738\\
\textbf{REG-PAML ($\gamma=10^{-2}$)} & \textbf{1.194\scriptsize{$\pm$ 0.014}}  & \emph{1.236\scriptsize{$\pm$ 0.011}}  &0.312\\
\textbf{PAML } & 1.471\scriptsize{$\pm$ 0.036}  & 1.569\scriptsize{$\pm$ 0.031} &0.038 \\
\textbf{AT-PAML } & 1.284\scriptsize{$\pm$ 0.005} &  1.382\scriptsize{$\pm$ 0.011} &0.018 \\
\bottomrule
\end{tabular}
}    \caption{%
    Ablations of different methods for MAjor Users (MAU) and Minor Users (MU) on MovieLens. The best results are highlighted in \textbf{bold} and the second-best results are in \emph{italic}. \textbf{Avg} means average. P-value is drawn by two-tailed student’s t-test between the minor user MSEs and major user MSEs over the same methods. 
    }
\label{abs_gamma}
\end{table}

\section{Acknowledgements}
This work was supported by Alibaba Group through Alibaba Innovative Research (AIR) Program and Alibaba-NTU Joint Research Institute (JRI), Nanyang Technological University, Singapore.

The authors would like to thank Suming Yu, Zhenyu Shi, Rundong Wang, Xinrun Wang, Feifei Lin, Aye Phyu Phyu Aung, Hanzao Chen, Ziwen Jiang, Yi Cao, Yufei Feng for their helps. The authors would also thank the team of MELU for their codes.
\bibstyle{abbrv}
\bibliography{PAML}

\newpage
\onecolumn

\appendix

\section{Omitted Proofs}\label{proof}
\setcounter{assumption}{0}

\subsection{Proofs of Lemmas \ref{distribution} and~\ref{bound}}\label{proof_bound}

\setcounter{lemma}{0}
\begin{lemma}
 \rs{Suppose that the loss function is defined as  $L = \sum\nolimits_{i=1}^2p_i(\theta_i- x_i)^2$ , where $x_1,x_2$ are the actual preferences for task $1$ and $2$, and $\theta_i = \theta-\alpha \nabla_{\theta} L_{\tau_{i}}(\theta)$, when $p_1 \geq p_2$, $x_1\geq x_2$, we have  $(\theta^*_1- x_1)^2 \leq (\theta^*_2 - x_2)^2 $, where $\theta^*$ is the optimal parameter $\theta$.}
\end{lemma}

\begin{proof}

Taking the gradient of $L$, we have:
\begin{align*}
&\nabla_{\theta} L = \nabla_{\theta} p_1 L_1 + \nabla_{\theta} p_2 L_2,\\
&  = p_1 \nabla_{\theta}\left(x_1 - \theta - \alpha \nabla_{\theta}(x_1 -  \theta)^ 2\right)^ 2 + p_2\nabla_{\theta}\left(x_2 - \theta -  \alpha \nabla_{\theta}(x_2 -  \theta)^ 2\right)^2,\\
&  = p_1 \nabla_{\theta}(x_1 -\theta -  2\alpha (\theta - x_1 ))^2 + p_2\nabla_{\theta}(x_2 -\theta -  2\alpha (\theta - \theta - x_2))^2,  \\
 & = 2 p_1 (x_1 - \theta -2\alpha (\theta - x_1) )(-2\alpha-1)+ 2p_2 (x_2 - 2\alpha (\theta - x_2) )(-2\alpha-1) .
  \end{align*}Solving $\nabla_{\theta^\prime} L = 0 $ above, we can get 
  \begin{align*}
& 2 p_1 (x_1 - \theta - 2\alpha (\theta - x_1) )(-2\alpha)+ 2p_2 (x_2 - \theta - 2\alpha (\theta- x_2) )(-2\alpha)  = 0 \\
&\rightarrow \theta^* = \frac{(2\alpha+1) (x_2p_2+x_1p_1)} {(2\alpha+1) (p_2+p_1)}  = \frac{(x_2p_2+x_1p_1)} {(p_2+p_1)}.
  \end{align*}If $p_1 \geq p_2$, putting $\theta^*$ into $L_1$ and $L_2$, we have $(\theta^*_1- x_1)^2\leq (\theta^*_2 - x_2)^2 $, which concludes our proof.
\end{proof}

\begin{lemma}
Based on Lemma \ref{distribution}, we further defined as $L^\prime =\sum\nolimits_{i=1}^2p_i(\theta^\prime_i- x_i)^2$, $\text { s.t. } {\theta}^\prime_{i}={\theta}-\alpha(h_i) \nabla_{{\theta}} L_{\tau_{i}}({\theta})$, ($x_1,x_2$ are the target values for task $1$ and $2$, where $x_2> x_1$) there exists $\alpha_1$ and $\alpha_2$ that satisfying
\begin{equation*}
(\theta^{*}_2- x_2)^2 \geq (\theta^{*\prime}_2 - x_2)^2 \text{and}\  L^*\geq L^{*\prime},
\end{equation*}
where $L^*$ and $L^{*\prime}$ are the optimal loss functions for standard MAML and adaptive MAML respectively. $\theta^{*\prime}$ is the optimal $\theta^{\prime}$ parameters.
\end{lemma}
\begin{proof}
  
  For $L^\prime$, using the similar approach, we have:
 \begin{align*}
&\nabla_{\theta} L^\prime = \nabla_{\theta} p_1 L^\prime_1 + \nabla_{\theta} p_2 L^\prime_2 \\
 & = 2 p_1 (x_1 - \theta^\prime - 2\alpha_1 (\theta^\prime - x_1) )(-2\alpha_1-1)+ 2p_2 (x_2 -\theta^\prime -  2\alpha_2 (\theta^\prime - x_2) )(-2\alpha_2-1)  
  \end{align*}
Solving $\nabla_{\theta} L = 0 $ above, we can get

  \begin{equation*}
\theta^{*\prime} = \frac{(2\alpha_1+1)^2p_1x_1+ (2\alpha_2+1)^2p_2x_2}{(2\alpha_1+1)^2 p_1 + (2\alpha_2+1)^2 p_2} 
\end{equation*}

Then if taking $\alpha_1= \alpha$ and $\alpha_2= \frac{[(2\alpha+1)\sqrt{\frac{q_1}{q_2}}-1]}{2}$. We can check that $(\theta^{*}_2- x_2)^2 \geq (\theta^{*^\prime}_2 - x_2)^2 $ and $ L^*\geq L^{*\prime}$.
  \end{proof}

\subsection{Proof of upper bound in Sec.~\ref{mar}}\label{ubound}
\begin{proposition}
Suppose $L_{\tau_i}(\boldsymbol{\theta})$ is Lipschitz continuity w.r.t user embedding, i.e., $\|L_{\tau_i}(\boldsymbol{\theta}) - L_{\tau_j}(\boldsymbol{\theta})\|\leq \kappa \|h_i -h_j \ \|$. The upper bound for the multi-user case is:
\begin{equation*}
       \sum\nolimits_{u_i,u_j\in\mathcal{U};i>j}||L_{\tau_i}\left(\boldsymbol{\theta}_{i}\right) -     L_{\tau_j}\left(\boldsymbol{\theta}_{j}\right)||\leq  \sum\nolimits_i\left(   (|\mathcal{U}|-1)||\frac{\partial L}{\partial\boldsymbol{\theta}}||^2|| \alpha(h_i)||\right),
\end{equation*}

\end{proposition}
\begin{proof}
We firstly consider a simple two-agent case: $ \min ||L_{\tau_i}\left(\boldsymbol{\theta}_{i}\right) -     L_{\tau_j}\left(\boldsymbol{\theta}_{j}\right) ||$. 

By Taylor series approximation, we have
\begin{small}
\begin{align*}
|| L_{\tau_i}\left(\boldsymbol{\theta}_{i}\right) -     L_{\tau_j}\left(\boldsymbol{\theta}_{j}\right) || \approx || L_{\tau_i}(\boldsymbol{\theta}) + \frac{\partial L_{\tau_i}}{\partial\boldsymbol{\theta}}(\boldsymbol{\theta}_i-\boldsymbol{\theta})- (L_{\tau_j}(\boldsymbol{\theta}) + \frac{\partial L_{\tau_j}}{\partial\boldsymbol{\theta}}(\boldsymbol{\theta}_j -\boldsymbol{\theta}))||.
\end{align*}
\end{small}
Here\footnote{With a little abuse of the notation, we use $\frac{\partial L_{\tau_i}}{\partial\boldsymbol{\theta}}$  to replace $\frac{\partial L_{\tau_i}^\mathrm{T}}{\partial\boldsymbol{\theta}}$.}, We use the first-order Taylor series to approximate the loss function $L_{\tau_i}$ and $L_{\tau_j}$. Then, through the triangle inequality, we have:
\begin{align*}
    &|| L_{\tau_i}\left(\boldsymbol{\theta}_{i}\right) -     L_{\tau_j}\left(\boldsymbol{\theta}_{j}\right) || \\ &\leq  || \frac{\partial L_{\tau_i}}{\partial\boldsymbol{\theta}}(\boldsymbol{\theta}_i-\boldsymbol{\theta})- \frac{\partial L_{\tau_j}}{\partial\boldsymbol{\theta}}(\boldsymbol{\theta}_j -\boldsymbol{\theta})|| + ||L_{\tau_i}(\boldsymbol{\theta}) -L_{\tau_j}(\boldsymbol{\theta})|| \\
    &=  ||  \frac{\partial L_{\tau_i}}{\partial\boldsymbol{\theta}}\alpha(h_i) \nabla_{\boldsymbol{\theta}} L_{\tau_{i}}(\boldsymbol{\theta})- \frac{\partial L_{\tau_j}}{\partial\boldsymbol{\theta}}\alpha(h_j) \nabla_{\boldsymbol{\theta}} L_{\tau_{j}}(\boldsymbol{\theta})||  + ||L_{\tau_i}(\boldsymbol{\theta}) -L_{\tau_j}(\boldsymbol{\theta})|| \\ %
    &\leq ||  \frac{\partial L_{\tau_i}}{\partial\boldsymbol{\theta}}\alpha(h_i) \nabla_{\boldsymbol{\theta}} L_{\tau_{i}}(\boldsymbol{\theta})- \frac{\partial L_{\tau_j}}{\partial\boldsymbol{\theta}}\alpha(h_j) \nabla_{\boldsymbol{\theta}} L_{\tau_{j}}(\boldsymbol{\theta})||  + ||h_i|| + ||h_j||,
\end{align*}
$\!$ For the first term, we have:
\begin{align*}
   & || \frac{\partial L_{\tau_i}}{\partial\boldsymbol{\theta}}\alpha(h_j) \nabla_{\boldsymbol{\theta}} L_{\tau_{i}}(\boldsymbol{\theta})- \frac{\partial L_{\tau_j}}{\partial\boldsymbol{\theta}}\alpha(h_i) \nabla_{\boldsymbol{\theta}} L_{\tau_{j}}(\boldsymbol{\theta})|| \\
   &= || (\frac{\partial L_{\tau_{i}}}{\partial\boldsymbol{\theta}})^2\alpha(h_i)-(\frac{\partial L_{\tau_{j}}}{\partial\boldsymbol{\theta}})^2\alpha(h_j) ||\\
   & \leq||\frac{\partial L_{\tau_{i}}}{\partial\boldsymbol{\theta}}||^2|| \alpha(h_i)|| + ||\frac{\partial L_{\tau_{j}}}{\partial\boldsymbol{\theta}}||^2|| \alpha(h_j)|| 
\end{align*}
The last inequality is obtained by the Cauchy--Schwarz inequality and the triangle inequality. The inequality above is the upper bound of the distance of losses for any two users. 

Now, we use the same approach to  the multi-user case:
$
   \sum\nolimits_{u_i,u_j\in\mathcal{U};i>j}||L_{\tau_i}\left(\boldsymbol{\theta}_{i}\right) -     L_{\tau_j}\left(\boldsymbol{\theta}_{j}\right)||,
$

Following the same approach, the upper bound for multi-user case is
\begin{equation*}
   \sum\nolimits_i(|\mathcal{U}|-1)||\frac{\partial L_i}{\partial\boldsymbol{\theta}}||^2|| \alpha(h_i)||  +|\mathcal{U}|(|\mathcal{U}|-1)\max_i||h_i||. 
\end{equation*}
This inequality is drawn by the sum of the inequality for all users. The coefficient $(|\mathcal{U}|-1)$ is obtained by the fact that for each loss $L_{\tau_i}$, we need to compute $(|\mathcal{U}|-1)$ times. The coefficient $|\mathcal{U}|*(|\mathcal{U}|-1)$ is by the fact that there are $2*\binom{|\mathcal{U}|}{2}$ number of $\max_i||h_i||$ in the whole calculation. Moreover, for each $\boldsymbol{\theta}$, the right term $\max_i||h_i||$ is a constant. Therefore, we only need to focus on the left term $\sum\nolimits_i(|\mathcal{U}|-1)||\frac{\partial L}{\partial\boldsymbol{\theta}}||^2|| \alpha(h_i)||$. Setting $|\mathcal{U}|*(|\mathcal{U}|-1)\max_i||h_i|| = C$, for the gradient descent methods, the loss function and its upper bound can be written as:
\begin{align*}
    \sum\nolimits_i L_{\tau_i}\left(\boldsymbol{\theta}_{i}\right) + \sum\nolimits_{u_i,u_j\in\mathcal{U};i>j}||L_{\tau_i}\left(\boldsymbol{\theta}_{i}\right) -     L_{\tau_j}\left(\boldsymbol{\theta}_{j}\right)|| \leq  \sum\nolimits_i\left( L_{\tau_i}\left(\boldsymbol{\theta}_{i}\right) + (|\mathcal{U}|-1)||\frac{\partial L_i}{\partial\boldsymbol{\theta}}||^2|| \alpha(h_i)||\right) + C.
\end{align*}
This completes the proof.
\end{proof}
\begin{remark}
REG-PAML can be explained as an approximated proximal  regularization.
\end{remark}
Recall the optimizer in~\citep{rajeswaran2019meta}:
\begin{equation*}
    L_i(\boldsymbol{\theta}_{i}) + \frac{\lambda}{2} ||\boldsymbol{\theta}_{i} - \boldsymbol{\theta}||^2,
\end{equation*}
where $\lambda$ is a constant parameter ( the regularization strength).  If we expnad the last term $||\boldsymbol{\theta}_{i} - \boldsymbol{\theta}||^2$ by the Taylor series, we have 
\begin{equation}
||\boldsymbol{\theta}_{i} - \boldsymbol{\theta}||^2 = ||\frac{\partial L_i}{\partial\boldsymbol{\theta}}\alpha(h_i)||^2 \leq ||\frac{\partial L_i}{\partial\boldsymbol{\theta}}||^2||\alpha(h_i)||^2,
\end{equation}
The last term is similar to our regularizer.  Therefore, our method can be also explained as an approximated proximal  regularization.
\section{Kernel-based PAML}\label{tree}

This section we discuss how kernel-based method can calculate similarity. Formally speaking, when a new user $u_i$ with embedding $h_i\in \mathbb{H}$ comes, our goal is to find the users with the most similar interest  to $u_i$. We first define a (oracle) classifier as $g(h):\mathbb{H}\rightarrow \mathbb{R}$, which maps the embedding to a certain group (e.g., the fishing enthusiast group). That is, the classifier $g(h)$ gives the same output for the users with similar interest. Then, the model can find which users are similar to the new user $u_i$. However, finding this classifier is not an easy task. Inspired by the Maximum Mean Discrepancy (MMD) approach~\citep{gretton2012kernel}, which converts the problem of finding the classifier into calculating the distance of probability, we build the objective as:
\begin{equation}\label{obj}
    s:=\inf\nolimits_{j\in U} ||\sup\nolimits _{g \in \mathcal{G}} \sum\nolimits_l g(h^l_i)-g(h^l_j)||^2,
\end{equation}
where $\mathcal{G} = \{g:||g||_H \leq 1\}$ is the set of classified functions ($\vert\vert\cdot\vert\vert_H$ is the norm function in Hilbert space $H$). $h^l$ is the value for $l$-th basis ($h = ( h^l)_{l=1}^M \in \mathbb{R}^M$, $M$ is the embedding size).  $|| \sup\nolimits _{g \in \mathcal{G}}\sum\nolimits_l g(h^l_i)-g(h^l_j)||$ is exactly the MMD~\citep{gretton2012kernel}, revealing the disparity between two distributions. %
In our setting, Eq.~\eqref{obj} can be explained as finding users with similar interest to the user $i$ in the embedding space if we regard each embedding component as random variable. %

With Riesz's representation theorem~\citep{scholkopf2002learning}, we have $g(h) =\langle g, \phi(h)\rangle$, where $\phi(h)$ is the feature space map from $\mathbb{H}$ to $H$. Applying similar approach from~\citep{gretton2012kernel}, we obtain:
\begin{equation*}
        s  = \inf\nolimits_{j\in U} || \sup\nolimits _{g \in \mathcal{G}} \langle g, \phi(h_i)-\phi(h_j)\rangle||
     = \inf\nolimits_{j\in U}  || \phi(h_i)-\phi(h_j)||_{\mathcal{H}},
\end{equation*}
 where $\langle \cdot,\cdot\rangle$ is the inner product. Here, $|| \phi(h_i)-\phi(h_j)||_{\mathcal{H}} = \langle \phi(h_i)-\phi(h_j), \phi(h_i)-\phi(h_j)\rangle = \langle \phi(h_i) - \phi(h_j) \rangle^{2}_{\mathcal{H}} := \kappa\left(h_i, h_{j}\right)$. Thus, we can calculate $s = \kappa\left(h_i, h_{j}\right)$ without knowing the \emph{classifier}, which can be used to estimate the similarity.

\begin{table*}[t]
    \centering
    \small
    {\scriptsize
\begin{tabular}{c|c|c|c}
\toprule
 & \textbf{MovieLens-1M} & \textbf{BookCrossing} & \textbf{Taobao dataset}\\
\midrule
Number of the selected users & 4832
 & 7827 & 6591\\
Selected item numbers & 3883 & 185973 &45067 \\
Selected features for items & genre, direct, actor & published year, publisher & ctr, cvr, price, logctr, logcvr, logprice, ctrcvr, cvrprice, ctrcvrprice.\\
Features for users & gender, age, occupation, zipcode & state, country, age & gender, age, occupation, zipcode, state, country \\

\bottomrule
\end{tabular}
}
    \captionsetup{justification=centering}
        \caption{Statistics of MovieLens-1M, BookCrossing, and Taobao dataset.}    \label{datasets}
\end{table*}

\section{Algorithms}\label{alg}
The complete pseudo-code can be found in Algorithm~\ref{alg1}. The searching and storing methods are in Algorithm~\ref{alg2}. Recall the loss functions and the gradient descents in the main text:

\begin{equation}\label{tloss}
   \mathcal{L} = \sum\nolimits_i L_{\tau_i} +  \gamma L^r_i,
\end{equation}

\begin{align}
\boldsymbol{\theta} = \boldsymbol{\theta} - \beta \sum\nolimits_{\tau_i}\nabla_{\boldsymbol{\theta}}L_{\tau_i}(\boldsymbol{\theta}_i)\Bigg( I - \alpha(h_i;\boldsymbol{\psi})\nabla^2_{\boldsymbol{\theta}} L_{\tau_{i}}(\boldsymbol{\theta})\Bigg),\label{meta1}\\
   \boldsymbol{\psi} = \boldsymbol{\psi}  + \beta \nabla_{\boldsymbol{\psi}} L_{\tau}(\boldsymbol{\theta}) \nabla_{\boldsymbol{\psi                 }}\alpha(h_i;\boldsymbol{\psi}) L_{\tau_{i}}(\boldsymbol{\theta}).\label{meta2}
\end{align}

\section{MORE DETAILS ABOUT THE
EXPERIMENTS} \label{md_exp}
  
     \subsection{Data preprocessing step}\label{preprocess}  The data pre-processing scheme is:  1) We rank the users with the number of their log data and choose the $80\%$ users least log data as the cold-start users (the value $80\%$ is based on the Pareto Principle). 2) We randomly split the users by $7:1:2$ for training, validation, and testing. 3) We tick out the users when they either have blank or wrong features (including the ages are less than 10 or larger than 100 as well as the location features with garble) or the items they viewed are less than a threshold (two items for rating prediction and two clicked items for CTR prediction). 4) We separate the support and query sets for each user with a ratio of 80$\%$:20$\%$ randomly. %
     
     \subsection{Hyper-parameters and Network Structures}\label{Hyper-parameters} The hyper-parameters for AT-PAML and REG-PAML can be found in Table~\ref{hyper}. 
     
     For the network structures, we would like to discuss the embedding module and the learning-rate module. The embedding layer is directly built from torch.nn.Embedding module\footnote{\url{https://pytorch.org/docs/master/generated/torch.nn.Embedding.html}}. we first use embedding layer for users' features and item features respectively and then we concatenate them together. Taking the MovieLens as an example, The code is as follows:

For the users' features:
\begin{small}
\begin{python}
    """
    user embedding.
    """
        self.embedding_gender = torch.nn.Embedding(num_embeddings=self.num_gender,
            embedding_dim=self.embedding_dim)

        self.embedding_age = torch.nn.Embedding(num_embeddings=self.num_age,
            embedding_dim=self.embedding_dim)

        self.embedding_occupation = torch.nn.Embedding(num_embeddings=
        self.num_occupation,
            embedding_dim=self.embedding_dim)

        self.embedding_area = torch.nn.Embedding(num_embeddings=self.num_zipcode,
            embedding_dim=self.embedding_dim)
        self.embedding_genre = torch.nn.Embedding(num_embeddings=self.num_genre,
            embedding_dim=self.embedding_dim)
        
\end{python}

\end{small}

For the item's features,

\begin{python}
    """
    item embedding.
    """
    self.embedding_genre = torch.nn.Embedding(num_embeddings=self.num_genre, 
    embedding_dim=self.embedding_dim)
\end{python}

Similar embedding layers are used for BookCrossing.

For the production dataset, due to the large corpus of items' and users' features of real-world dataset, it is better to extract some useful features rather than putting them into the network together. Here, we use the based ranker to generate the item/user features~\citep{lin2019pareto}. Then, we directly leverage these features as the input of our models.

For the LR modules. During the experiments, we build the FC layer as:

\begin{small}
\begin{python}
    def forward(self, x):
        user_emb, item_emb = x
        user_emb_exp = user_emb.repeat(item_emb.shape[0],1)
        x = torch.cat((item_emb.float(), user_emb_exp.float()), 1)

        # for lr
        lremb = self.fc1_1(user_emb_exp.float())
        lremb = F.relu(lremb)
        lremb = self.fc1_2(lremb)
        lr_local_raw = 0.001 * torch.sigmoid(self.fc_lr(self.lremb))
\end{python}
\end{small}
which shrinks the learning rate. 

For AT-PAML, we introduce the warm-up (one epoch) stage to let the model store some embeddings and value first, the LRs in warm-up stage is 0.0005.

The \textbf{network structures} for MeLU, transfer-learning, Meta-SGD are similar to PAML except that PAML has an extra LR module.  

The \textbf{hyper-parameters} for MeLU, transfer-learning, Meta-SGD are also similar to PAML except that MeLU has an fixed inner LR ($10^{-5}$). 

We provide the Python code for the search and calculate methods to help our reader better understand our model. Some of the codes are from\footnote{\url{https://github.com/mjacar/pytorch-nec}}.
\begin{small}
\begin{python}
from pyflann import FLANN
self.kdtree = FLANN()

def search_and_calculate(self, lookup_key):
    """
    A brief description for search function and kernel calculating function.
    """
    lookup_indexes = self.kdtree.nn_index(
        lookup_key.data.cpu().numpy(), min(self.num_neighbors, len(self.keys)))[0][0]
    output = 0
    kernel_sum = 1e-5
    for i, index in enumerate(lookup_indexes):
      kernel_val = self.kernel(self.keys[int(index)], lookup_key[0])
      output += kernel_val * self.values[int(index)]
      kernel_sum += kernel_val
    output = output / kernel_sum
    return output
\end{python}
\end{small}

\begin{table}[H]
\centering
\caption{Hyper-parameters for PAML, AT-PAML and REG-PAML.  d is 5 for Movielens, 10 for BookCrossing, and 2 for production dataset.}
\label{hyper}
{
\begin{tabular}{|c|c|c|c|c|}
\toprule
\textbf{Hyper-paramter} & \textbf{AT-PAML} &\textbf{REG-PAML} &\textbf{PAML}\\
\midrule
Embedding size for each feature& 32 & 32 & 32\\
Decision layer size (from first to last) & 320,192,d & 320,192,d & 320,192,d\\
Learning rate layer size (from input to output) & 64,32,1 & 64,32,1 & 64,32,1\\
Outer learning rate $\beta$ & 5e-5 & 5e-5 &5e-6\\
Epochs & 20 & 20 & 20\\
Batch size & 32 &32 &32\\
Tree Memory size & 10000  & None& None\\
Warm up epoch &1 &None & None\\
Warm up LR &5e-6 &None & None\\
Number of neighbours for inference  & 5 & None & None\\
regularizer parameter $\gamma_0$ & None & 1e-3 & None\\
 $\delta$ &2 & None & None\\
  The number of nearest users $K$ &20 & None & None\\
\bottomrule
\end{tabular}
}
\end{table}

     \subsection{Details of datasets.} \label{dataset}
     The details of different datasets  can be found in Table~\ref{datasets}. 
    
     \subsection{Computational resources and platforms}\label{platform} We use Pytorch\footnote{\url{https://pytorch.org/}} as the deep learning library and we use the ADAM as the optimizer~\citep{kingma2014adam}. All the experiments are done on a single GeForce GTX 1080 Ti. The production dataset is obtained by the company's recommendation pipeline. The package for the two-tailed student's t-test is from\footnote{\url{https://docs.scipy.org/doc/scipy/reference/generated/scipy.stats.ttest_ind.html}}. 
     
    \subsection{Loss function}\label{loss} For the rating prediction tasks, we leverage the mean squared loss $\mathcal{L}_{\tau_i} = \frac{1}{\left|\tau_i\right|} \sum\nolimits_{j \in \tau_i}\left(y_{i j}-\hat{y}_{i j}\right)^{2}$, where $y_{i j}$ is item $j$'s rating score predicted by user $i$'s prediction model and $\hat{y}_{i j}$ is the user $i$'s actual rating score for item $j$. For the CTR prediction tasks, we take advantage of the negative entropy loss $\mathcal{L}_{\tau_i} = \frac{1}{\left|\tau_i\right|} \sum\nolimits_{j \in \tau_i}-\omega_j\hat{y}_{i j}\log y_{i j}$, where $\omega_j$ is the weight for each label to handle the imbalance issue between negative (not click) and positive (click) samples. Furthermore, for CTR prediction, since the negative samples are far more than positive samples in real-world dataset, a weighted negative entropy loss is introduced to guide the model to pay more attention to the positive samples.

\subsection{Evaluation Metrics}~\label{Metrics}
We introduce MSE, Average nDCG@K, AUC, and Average NEL as our evaluation metrics. Here are the details:
\begin{enumerate}
    \item \textbf{MSE.} MSE is calculated by \\ $\tiny \frac{1}{\left|D_{test}\right|}\sum\nolimits_{i \in D_{test}} \frac{1}{\left|\tau_i\right|}\sum\nolimits_{j \in \tau_i}\left(y_{i j}-\hat{y}_{i j}\right)^{2}$,where $D_{test}$ is the test dataset.
    \item \textbf{Average nDCG@K.} Average nDCG@K is calculated by \\$\frac{1}{|D_{test}|} \sum\nolimits_{i \in D_{test}} \frac{D C G_{K}^{i}}{I D C G_{K}^{i}}$,where $D C G_{K}^{i}=\sum\nolimits_{r=1}^{K} \frac{2^{R_{i r}}-1}{\log _{2}(1+r)}$, $R_{i r}$ are the true rating of the user $i$ for the $r$-th ranked item and $IDCG^i_K$ is the best possible $DCG^i_K$ for user $i$.
    \item \textbf{AUC.} AUC is a widely adopted metric to evaluate the performance of the classification task. We use it to validate our method in the CTR prediction task.
    \item \textbf{Average NEL.} Average NEL is defined as $\frac{\sum\nolimits_{i\in D_{test}}}{|D_{test}|} \frac{1}{\left|\tau_i\right|} \sum\nolimits_{j \in \tau_i}\allowbreak-\omega_j\hat{y}_{i j}\log y_{i j}$. We set $\omega_j$ as $0.1$ for the non-clicked item and $0.9$ for the clicked item. We also use it to validate our method in the CTR prediction task.
\end{enumerate}

\subsection{Extra Experiments and analysis}~\label{Experiments}
From Tab~\ref{abs_three}, the REG-PAML with $\gamma = 10^{-2}$ performs better than REG-PAML with $\gamma = 10^{-3}$
on average but we still choose REG-PAML with $\gamma = 10^{-3}$ as our standard method since it performs more robustly (lower variance values) than REG-PAML with $\gamma = 10^{-2}$ (0.021 v.s. 0.013). 

From Fig.~\ref{all_books} and~\ref{all_movies}, we find that both the AT-PAML and the REG-PAML have lower first quantile, mean and third quantile in MSE. Fig.~\ref{lrs2} indicates that the bigger $\gamma$ is, the more difference of LRs between the major and minor users is, revealing that the regularizer does have the ability to let the model focus on both major and minor users. From Fig. \ref{abs_vi}, we discover that different $\gamma$ values and tree-based module affects the users with mean MSE rather than users with the best or the worst MSE values. Moreover, from Fig.~\ref{abs_ls}, we find that our methods still perform well in both the minor and major users comparing with SOTAs.

Moreover, from Fig.~\ref{tsne}, we can find that our methods do have the ability to distinguish the minor users and the major users. Although there are some mistakes, this is because the current criteria for finding the minor users may not be accurate enough.

Fig.~\ref{fail} is a failure case, which reveals that when the model fails to perform well, it is also unable to distinguish the major and minor users. This indicates that the adaptive learning rate does have a connection with the performance of both major and minor users.

\begin{figure}[t]
\vspace{-0.2cm}
\subfigure[The local enlarged box plot for the distribution of MSE for each user in the test dataset in BookCrossing with filter.]{
\includegraphics[width=0.45\textwidth]{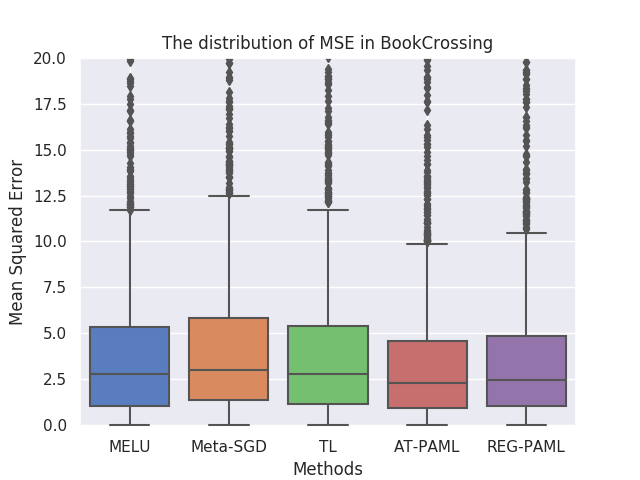}  %

}
\hfill 
\subfigure[The original box plot for the distribution of MSE for each user in the test dataset in BookCrossing.]{
\includegraphics[width=0.45\textwidth]{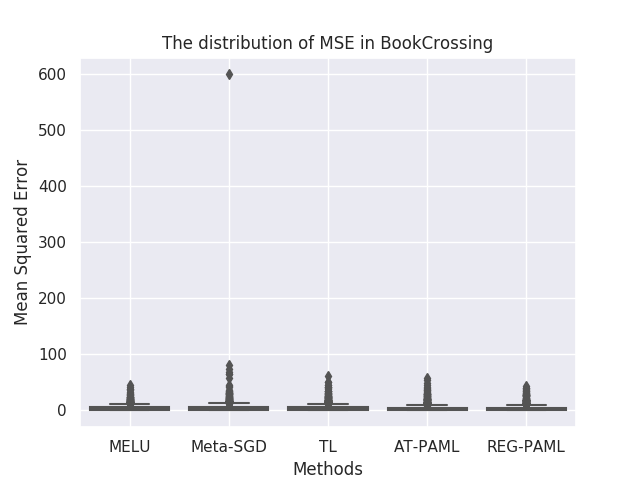} 

}
\caption{The Statistical analysis of different methods in BookCrossing.}\label{all_books}
\end{figure}

\begin{figure}[t]
\vspace{-0.2cm}
\subfigure[The local enlarged box plot for the distribution of MSE for each user in the test dataset in MovieLens.]{
\includegraphics[width=0.45\textwidth]{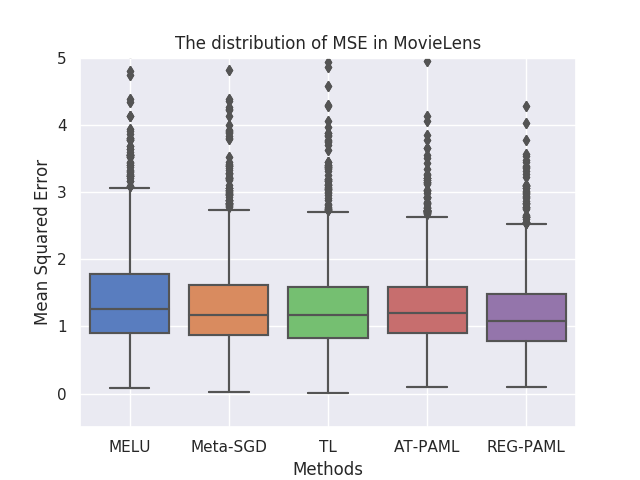}  %

}
\hfill 
\subfigure[The original box plot for the distribution of MSE for each user in the test dataset in MovieLens.]{
\includegraphics[width=0.45\textwidth]{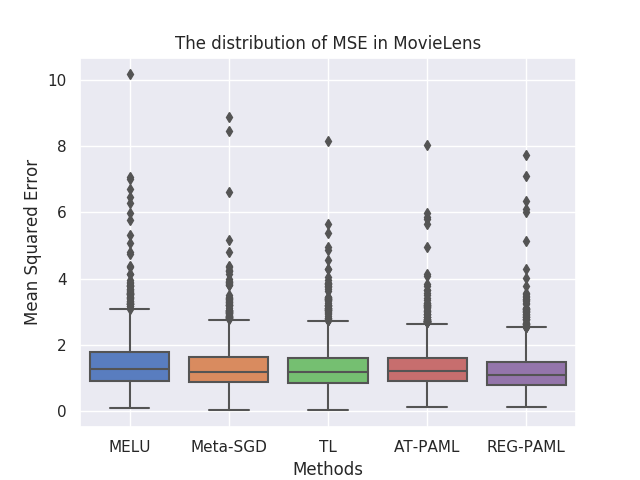} 

}
\caption{The Statistical analysis of different methods in MovieLens.}\label{all_movies}
\end{figure}

\begin{figure*}[tb!]
\centering
\subfigure[The distribution of minor and major user LRs for REG-PAML ($\gamma = 10^{-3}$) in Movielen.]{
\includegraphics[width=0.45\textwidth]{supple/sgd1e3_lrs_all.png}}\hfill
\subfigure[The distribution of LRs for REG-PAML ($\gamma = 10^{-3}$) in MovieLens.]{
\includegraphics[width=0.40\textwidth]{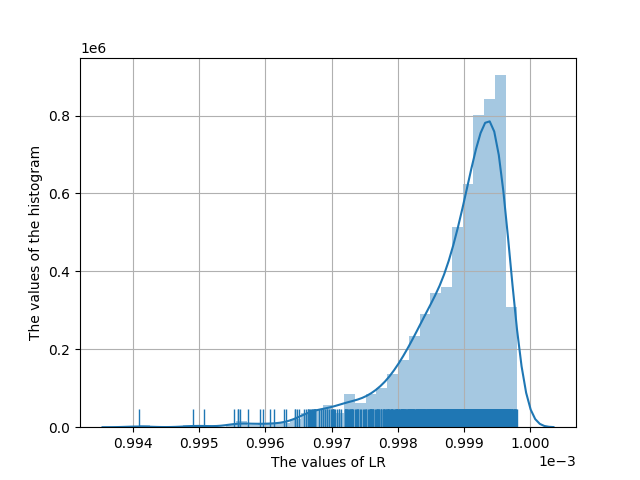}}
\hfill
\subfigure[The distribution of minor and major user LRs for REG-PAML ($\gamma = 10^{-5}$) in MovieLens.]{
\includegraphics[width=0.45\textwidth]{supple/sgd1e5_lrs_all.png}}\hfill 
\subfigure[The distribution of LRs for REG-PAML ($\gamma = 10^{-5}$) in MovieLens.]{
\includegraphics[width=0.45\textwidth]{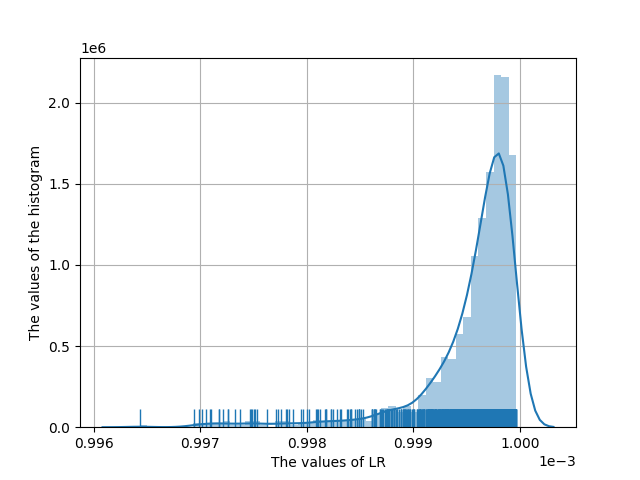}}\hfill 
\subfigure[The distribution of minor and major user LRs for PAML in MovieLens.]{
\includegraphics[width=0.45\textwidth]{supple/paml_lrs_all.png}}\hfill 
\subfigure[The distribution of LRs for PAML in MovieLens.]{
\includegraphics[width=0.45\textwidth]{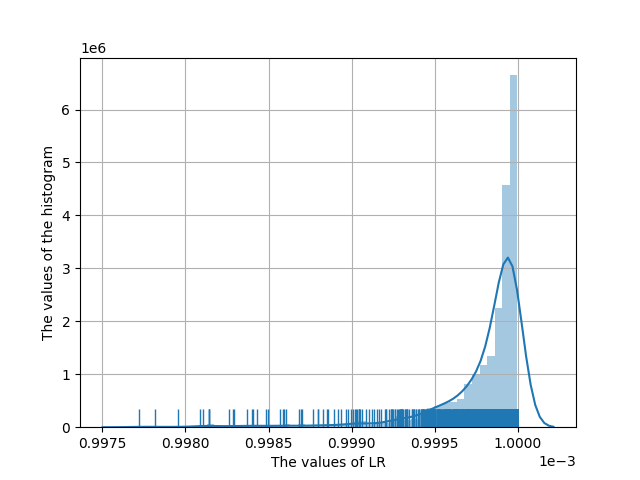}}\hfill

\caption{The visualization of LRs for REG-PAML ($\gamma = 10^{-3}$), REG-PAML ($\gamma = 10^{-5}$), and PAML. } \label{lrs2}
\end{figure*}

\begin{figure*}[tb!]
\centering
\subfigure[The distribution of LRs for AT-PAML in Movielens.]{
\includegraphics[width=0.45\textwidth]{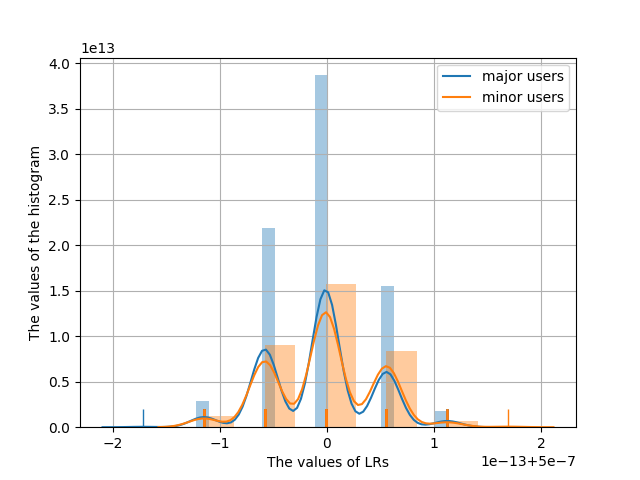}}\hfill
\subfigure[The distribution of minor and major user LRs for AT-PAML in MovieLens.]{
\includegraphics[width=0.45\textwidth]{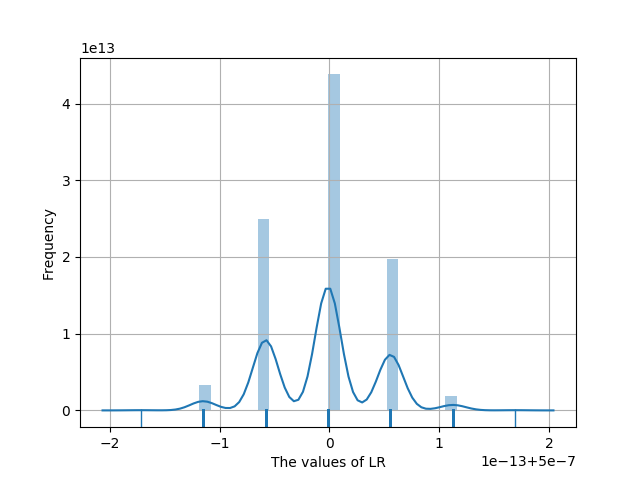}}\hfill
    \vspace{-0.1in}
\caption{A failure case. In the failure case, except for the larger (outer) learning rate ($5e-4$), other hyper-parameters are the same as AT-PAML. The results are: MSE: $2.127$, major MSE: $1.477$, minor MSE: $3.148$, nDCG@3: $0.782$, nDCG@5: $0.792$.}\label{fail}
\end{figure*}

\begin{figure*}[tb!]
\centering
\subfigure[The violin box of the MSEs for different methods in test dataset.]{
\includegraphics[width=0.45\textwidth]{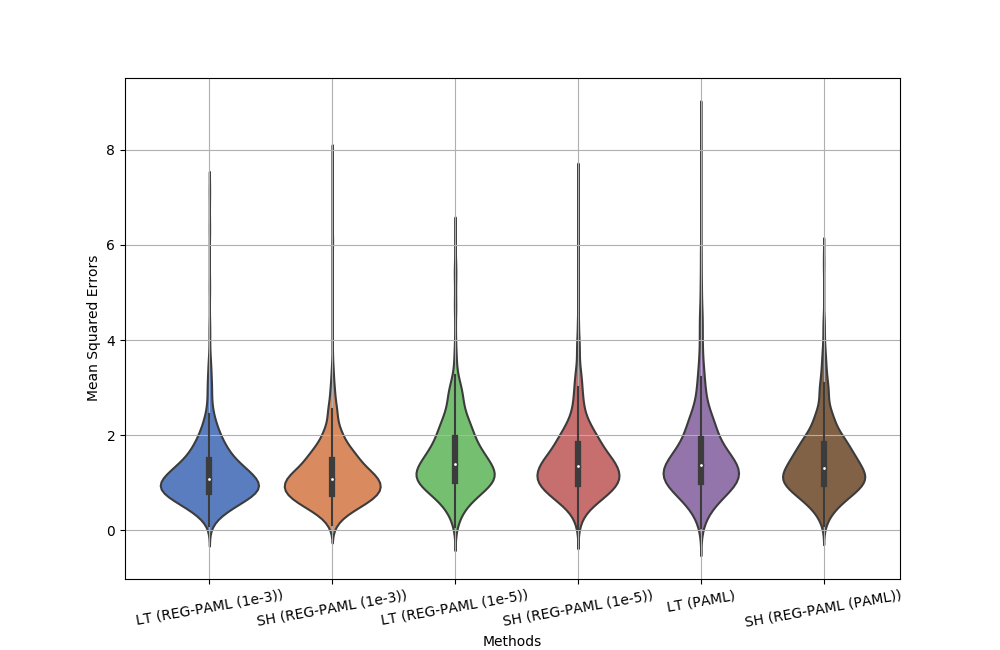}}\hfill
\subfigure[The local enlarged violin box of the MSEs for different methods in test dataset.]{
\includegraphics[width=0.45\textwidth]{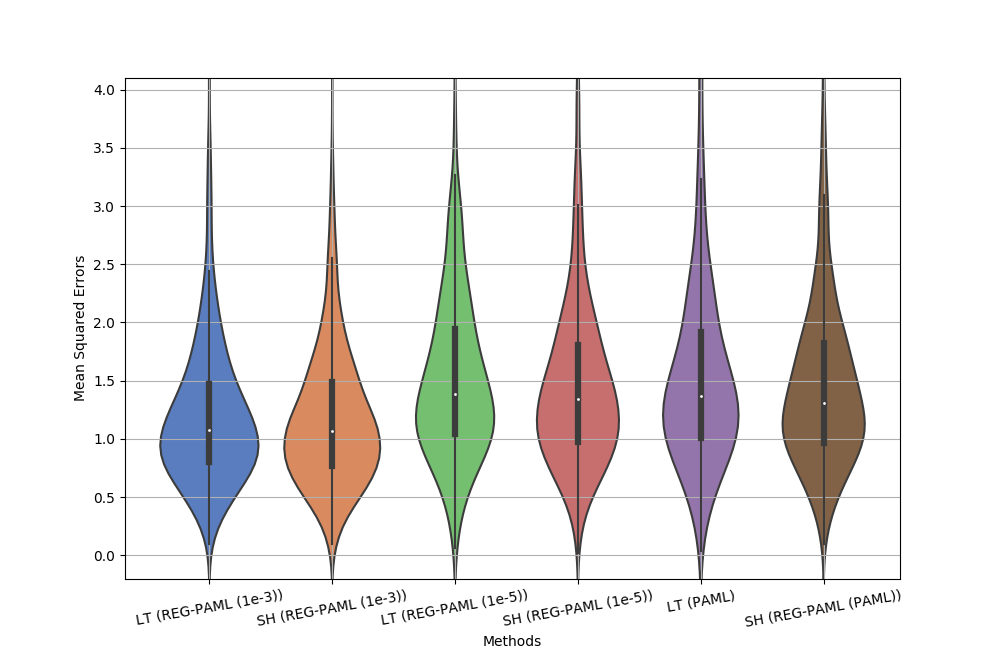}}\hfill
\caption{The performance and statistics of ablation methods in MovieLens. LT means the minor users, SH is the major users, and TL is the abbreviation for transfer learning. Higher is better. }\label{abs_vi}
\end{figure*}

\begin{figure*}[tb!]
\centering
\subfigure[The box plot of the MSEs for different methods in test dataset. ]{
\includegraphics[width=0.45\textwidth]{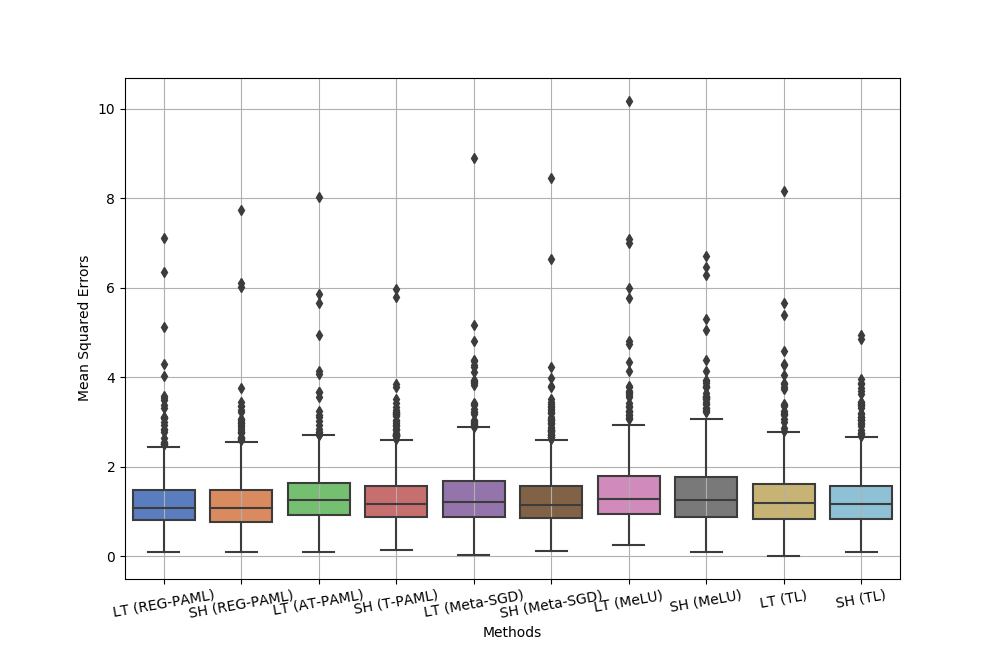}}\hfill
\subfigure[The local enlarged violin box of the MSEs for different methods in test dataset.]{
\includegraphics[width=0.45\textwidth]{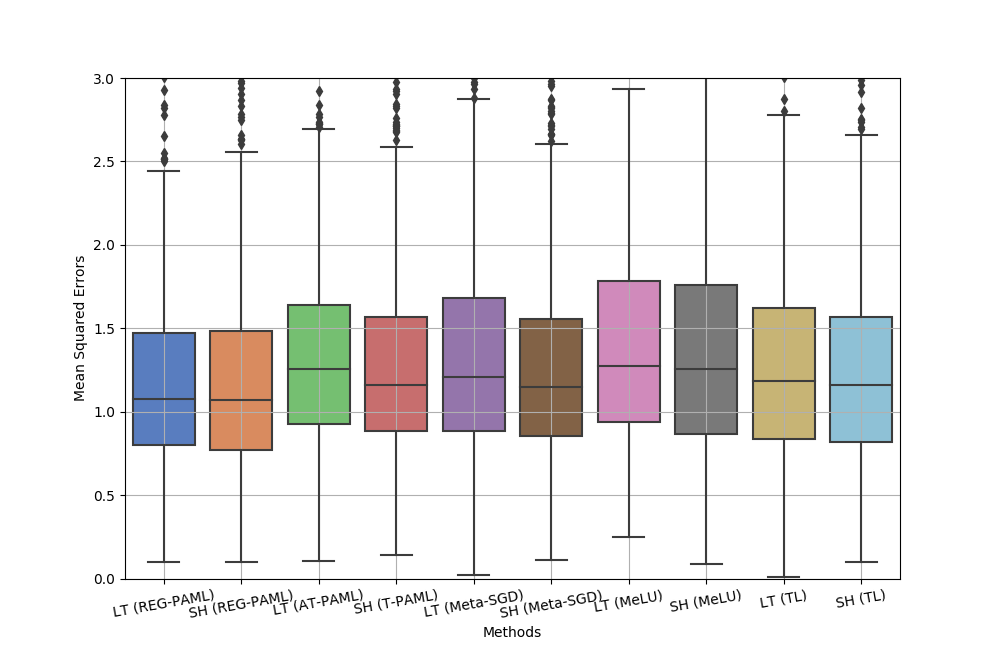}}\hfill
\caption{The performance and statistics of different methods in MovieLens. LT means the minor users, SH is the major users, and TL is the abbreviation for transfer learning. Higher is better.}\label{abs_ls}
\end{figure*}

\begin{figure*}[tb!]
\centering
\subfigure[The t-SNE visualization of REG-PAML ($\gamma = 10^{-3}$). ]{
\includegraphics[width=0.45\textwidth]{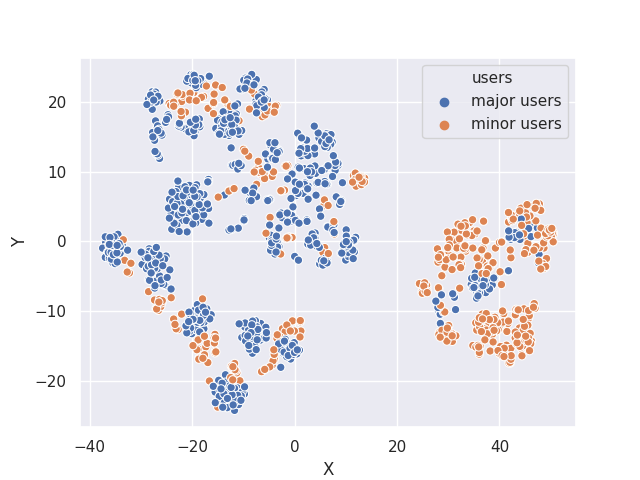}}\hfill 
\subfigure[The t-SNE visualization of PAML.]{
\includegraphics[width=0.45\textwidth]{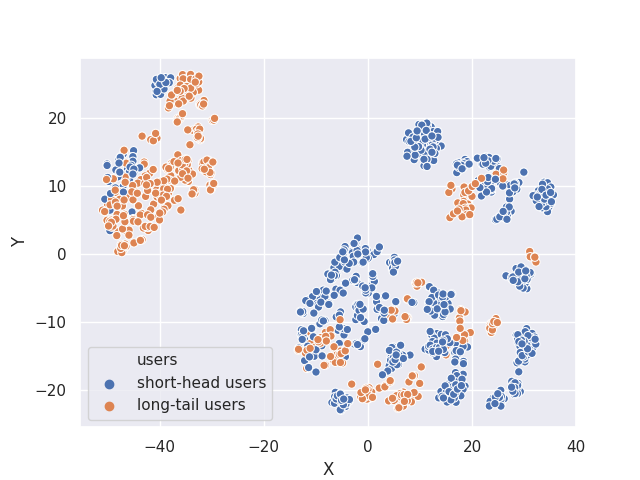}}\hfill 
\caption{The t-SNE visualization of user embeddings in MovieLens. X and Y are different dimensions. The blue dots are for major users while the orange dots are for minor users.}\label{tsne}
\end{figure*}

\end{document}